\documentclass[onecolumn,draftclsnofoot,12pt]{IEEEtran}

\usepackage{enumerate}

\usepackage{amsmath,amsthm}
\usepackage[noend]{algpseudocode}
\usepackage{float}
\usepackage{hyperref}
\usepackage{color}
\usepackage{makeidx}
\usepackage{bbm}
\usepackage{graphicx}
\usepackage{lipsum}
\usepackage{soul}
\usepackage{tabularx}
\usepackage{dsfont}

\usepackage{amsfonts}
\usepackage{times}
\usepackage{graphicx}
\usepackage{latexsym}
\usepackage{dsfont}
\usepackage{amssymb}
\usepackage{cite}
\usepackage{verbatim}
\usepackage{subfigure}

\newtheorem{theorem}{Theorem}

\newtheorem{assumption}{Assumption}

\newtheorem{corollary}[theorem]{Corollary}

\newtheorem{proposition}{Proposition}

\newcommand{\figref}[1]{{Fig.}~\ref{#1}}

\def\bb0{{\mathbb{0}}}

\def\ba{{\mathbf{a}}}
\def\bb{{\mathbf{b}}}

\def\bff{{\mathbf{f}}}
\def\bg{{\mathbf{g}}}
\def\bh{{\mathbf{h}}}

\def\bm{{\mathbf{m}}}
\def\bn{{\mathbf{n}}}

\def\bp{{\mathbf{p}}}

\def\bt{{\mathbf{t}}}

\def\bx{{\mathbf{x}}}
\def\by{{\mathbf{y}}}

\def\b0{{\mathbf{0}}}

\def\bA{{\mathbf{A}}}

\def\bI{{\mathbf{I}}}

\def\bR{{\mathbf{R}}}

\def\bbE{{\mathbb{E}}}

\def\cA{\mathcal{A}}

\def\cN{\mathcal{N}}

\def\sf0{{\mathsf{0}}}

\usepackage{epstopdf}

\newcommand{\sref}[1]{{Section}~\ref{#1}}

\DeclareMathOperator*{\argmax}{arg\,max}

\def\rm{\mathrm}

\begin{document}
\title{Deep Learning for mmWave Beam and Blockage Prediction Using Sub-6GHz Channels} 
\author{Muhammad Alrabeiah and Ahmed Alkhateeb \thanks{Muhammad Alrabeiah and Ahmed Alkhateeb are with Arizona State University (Email: malrabei, aalkhateeb@asu.edu).}}
\maketitle

\begin{abstract}
	Predicting the millimeter wave (mmWave) beams and blockages using sub-6GHz channels has the potential of enabling mobility and reliability in scalable mmWave systems. These gains attracted increasing interest in the last few years. Prior work, however, has focused on extracting spatial channel characteristics  at the sub-6GHz band first and then use them to reduce the mmWave beam training overhead. This approach has a number of limitations: (i) It still requires a beam search at mmWave, (ii) its performance is sensitive to the error associated with extracting the sub-6GHz channel characteristics, and (iii) it does not normally account for the different dielectric properties at the different bands. In this paper, we first prove that under certain conditions, there exist mapping functions that can predict the optimal mmWave beam and correct blockage status directly from the sub-6GHz channel, which overcome the limitations in prior work. These mapping functions, however, are hard to characterize analytically which motivates exploiting deep neural network models to learn them. For that, we prove that a large enough neural network can use the sub-6GHz channel to directly predict the optimal mmWave beam and correct blockage status with success probabilities that can be made arbitrarily close to one. Then, we develop an efficient deep learning model and evaluate its  beam/blockage prediction performance using the publicly available DeepMIMO dataset. The results show that the proposed solution can predict the mmWave blockages with more than 90$\%$ success probability.  Further, these results confirm the capability of the proposed deep learning model in predicting the optimal mmWave beams and approaching the optimal data rates, that assume perfect channel knowledge, while requiring no beam training overhead. This highlights the promising gains of leveraging deep learning models to predict  mmWave beams and blockages using sub-6GHz channels.   
\end{abstract}

\section{Introduction} \label{intro}
Enabling the high data rate gains of millimeter wave (mmWave) communications is very challenging in systems with high mobility and strict reliability constraints. This is mainly due to (i) the large training overhead associated with adjusting the beamforming vectors of the large mmWave arrays and (ii) the high sensitivity of mmWave signal propagation to blockages \cite{HeathJr2016,Andrews2016,Alkhateeb2018}. These challenges are expected to just get harder as wireless communication systems continue to move higher in frequency and deploy larger antenna arrays \cite{Rappaport2019,Sanguinetti2019}. Future mmWave communication systems, however, are expected to operate in multiple bands (including sub-6GHz and mmWave bands) \cite{Parkvall2017,Gonzalez-Prelcic2017}. Different than mmWave, sub-6GHz channels can generally be acquired with low training overhead---ideally with a single uplink pilot in time-division duplexing systems. Further, the propagation of the sub-6GHz signals is more robust to blockages. Then, can we leverage the spatial correlation between the sub-6GHz and mmWave channels to help reduce the high mmWave beam training overhead and maintain reliable links with blockages? More specifically, \textbf{can we use the sub-6GHz channels to directly predict the mmWave beams and blockages? Answering this interesting question is the goal of this paper.} 

\subsection{Related Work}

The general idea of using some knowledge about the sub-6GHz channels to aid the system and network operation at mmWave is motivated by the spatial correlation between the two bands, which has been verified through experimental measurements \cite{Hashemi2018, Peter2016, Nitsche2015}. On the network perspective, \cite{Hashemi2018} proposed a network architecture the leveraged the spatial correlation between sub-6GHz and mmWave bands for traffic scheduling and training overhead reduction. In \cite{Polese2017}, a dual connectivity protocol was developed that relies on a local coordinator to hand over the users between the two bands to avoid link failures. Leveraging deep learning tools, \cite{Burghal2019, Mismar2019} proposed strategies that learns the correlation between the sub-6GH and mmWave bands and exploits that for selecting the communication band or handing over the users from one band to the other. While the work in \cite{Hashemi2018, Semiari2017, Polese2017, Burghal2019, Mismar2019} is relevant, it does not target predicting the mmWave beams or blockages using sub-6GHz channels, which is the goal of this paper. 

To reduce the mmWave beam training overhead, \cite{Nitsche2015} designed a novel algorithmic framework to leverage the sub-6GHz spatial information in estimating the candidate mmWave beam directions. The feasibility of this solution was also studied in \cite{Nitsche2015} using a proof-of-concept prototype. This solution, however, was mainly limited to detecting the line-of-sight (LOS) mmWave direction. In \cite{Ali2018}, the spatial information from sub-6GHz was used to guide the compressive sensing based beam selection at mmWave bands and reduce the beam search overhead. With the same goal,  \cite{Ali2019} proposed an approach that constructs the mmWave channel covariance using the spatial characteristics extracted from the sub-6GHz band. This mmWave covariance knowledge can then be exploited to reduce the training overhead associated with the design of the analog or hybrid analog/digital precoding matrices.

While the interesting solutions in \cite{Nitsche2015,Ali2018,Ali2019} have the potential of reducing the search space of the mmWave beams, they share the following common limitations. First, the solutions in \cite{Nitsche2015,Ali2018,Ali2019} generally rely on the approach of estimating some spatial parameters, such as the angular characteristics and path gains, at the sub-6GHz band and then leverage them at mmWave. This makes their performance very sensitive to the parameters estimation error at the low-frequency bands. Also, this approach does not incorporate how the materials' dielectric coefficients, such as the reflection/scattering coefficients, differ in the two bands, which could be critical for the accurate modeling of the mmWave signal propagation. Further, the solutions in \cite{Ali2018,Ali2019} still requires relatively large beam training overhead at the mmWave band, which scales with the number of antennas. Finally, and to the best of our knowledge, no prior work has provided any theoretical guarantees on using the sub-6GHz channels to directly find the \textit{optimal} mmWave beams or detect the mmWave blockages.   

\subsection{Contribution}

This paper considers dual-band systems where the base station and mobile users employ both sub-6GHz and mmWave transceivers. It establishes the theoretical conditions under which sub-6GHz channels can be used to directly predict mmWave beams and blockages, and shows that deep learning models can be efficiently leveraged to achieve these objectives. The main contributions of this paper can be summarized as follows:
\begin{itemize}
	\item We prove that for any given environment, there exists a mapping function that can predict the optimal mmWave beam (out of a codebook) directly from the sub-6GHz channel if certain conditions are satisfied. These mapping functions, however, are hard to characterize analytically which motivated leveraging deep learning models to learn them. 
	
	\item Leveraging the universal approximation theory \cite{Hornik1989}, we prove that large enough neural networks can learn how to predict the optimal mmWave beams directly from sub-6GHz channel vectors with a success probability that can be made arbitrarily close to one.

	\item We show that a similar result can be established for blockage prediction, and identify the conditions under which the sub-6GHz channels can be used to predict whether or not the mmWave LOS link is obstructed. We also prove that large enough neural networks can be exploited to learn this blockage prediction with an arbitrarily high success probability.
	
	\item We propose a deep neural network model that efficiently uses the sub-6GHz channels to predict the optimal mmWave beams and blockage status. The developed model leverages transfer learning  to reduce the learning time overhead.	
\end{itemize}  
The proposed deep learning based mmWave beam and blockage prediction solutions were evaluated using the publicly-available dataset DeepMIMO \cite{DeepMIMO}. This dataset generates sub-6GHz and mmWave channels using the accurate 3D ray-tracing simulator Wireless InSite \cite{Remcom} which incorporates the materials' dielectric properties at the two bands.  The simulation results confirm the promising capability of deep learning models in learning how to predict the mmWave beams and blockages using sub-6GHz channels, as explained in detail in \sref{sec:Results}.

\textbf{Notation}: We use the following notation throughout this paper: $\bA$ is a matrix, $\ba$ is a vector, $a$ is a scalar, $\cA$ is a set of scalars, and $\boldsymbol{\mathcal{A}}$ is a set of vectors. $\|\ba \|_p$ is the p-norm of $\ba$. $|\bA|$ is the determinant of $\bA$, whereas $\bA^T$, $\bA^{-1}$ are its transpose and inverse. $\bI$ is the identity matrix.  $\cN(\bm,\bR)$ is a complex Gaussian random vector with mean $\bm$ and covariance $\bR$.


\section{System and Channel Models} \label{sec:System}

Consider the system model in \figref{fig:System} where a base station (BS) is communicating with one mobile user. The BS is assumed to employ two transceivers; one transceiver is working at sub-6GHz and employs $M_\text{sub-6}$ antennas, and the other one is operating at a mmWave frequency band and adopts an  $M_\text{mmW}$-element antenna array. For simplicity, we assume that the two antenna arrays belonging to the mmWave and sub-6GHz transceivers are co-located. As will be discussed in \sref{secIII}, however, the proposed concepts in this paper can be  extended to other setups with separated and distributed arrays. The mobile user is assumed to employ a single antenna at both mmWave and sub-6GHz bands. Next, we summarize the  system operation and the adopted channel model.

\begin{figure}[t]
	\centering
	\includegraphics[width=.9\linewidth]{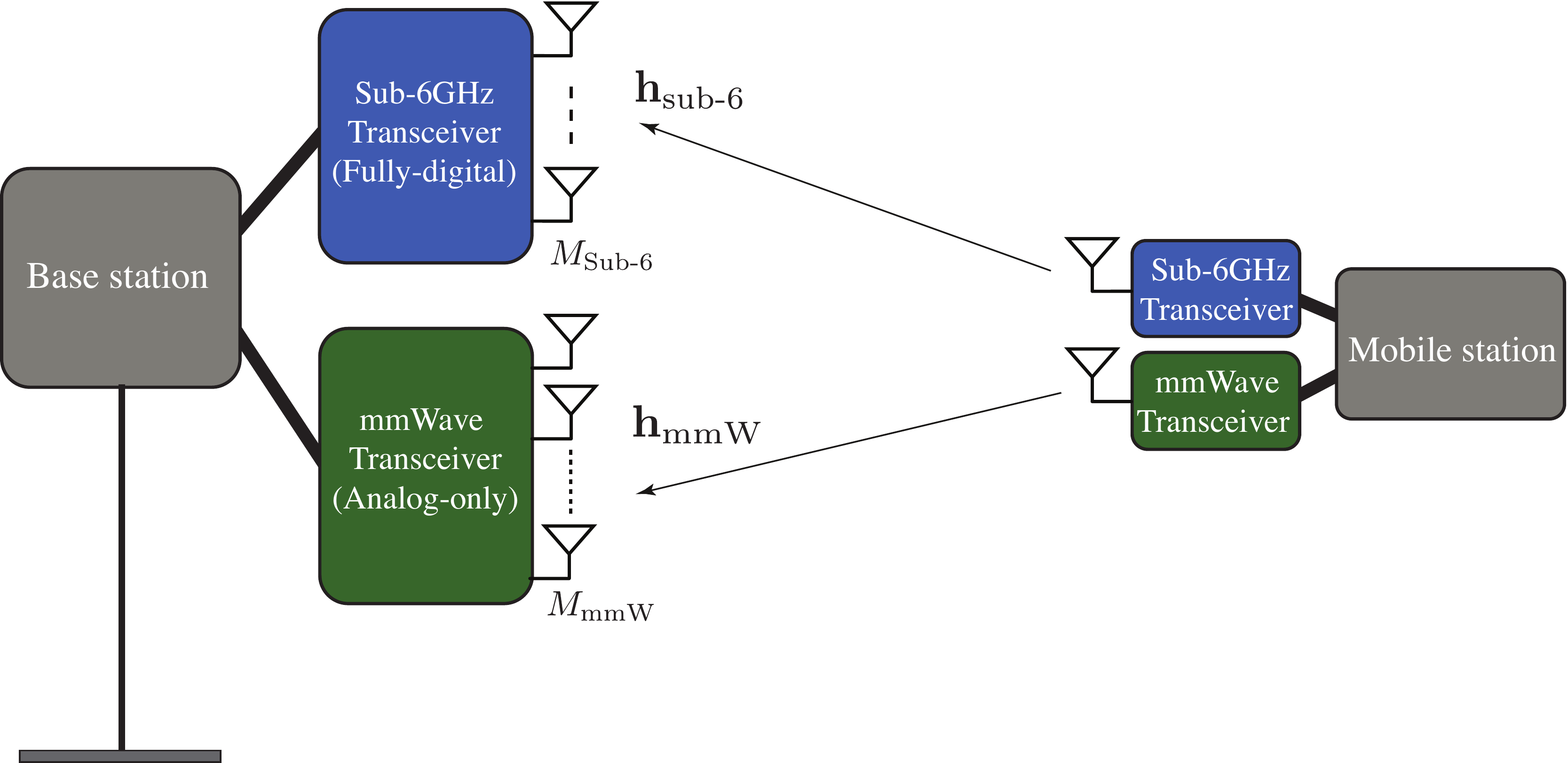}
	\caption{The adopted system model where a base station and a mobile user communicate over both sub-6GHz  and mmWave bands. The basestation and mobile user are assumed to employ co-located sub-6GHz and mmWave arrays. }
	\label{fig:System}
\end{figure}

\textbf{System Operation:} In this paper, we consider a system operation where the uplink signaling happens at the sub-6GHz band while the downlink data transmission occurs at the mmWave band. If $\bh_\text{sub-6}[k] \in \mathbb{C}^{M_\text{sub-6} \times 1}$ denotes the uplink channel vector from the mobile user to the sub-6 GHz BS array at the $k$th subcarrier,  $k=1,...,K$, then the uplink received signal at the BS sub-6GHz array can be written as 
\begin{equation}
\by_\text{sub-6} [k]=\bh_\text{sub-6}[k] s_\text{p}[k] + \bn_\text{sub-6}[k],
\end{equation}
where $s_\text{p}[k]$ represents the uplink pilot signal that satisfies  $\bbE\left|s_\text{p}[k]\right|^2=\frac{P_{\text{sub-6}}}{K}$, with $P_{\text{sub-6}}$ denoting the uplink transmit power from the mobile user. The vector $\bn_\text{sub-6} [k] \sim \mathcal{N}_\mathbb{C}\left(\boldsymbol{0}, \sigma^2 \boldsymbol{I}\right)$ is the receive noise at the BS sub-6GHz array. The sub-6 GHz transceiver is assumed to employ a fully-digital architecture, which allows for the channel estimation process to be done in the baseband.  

For the downlink transmission, the BS employs the mmWave transceiver. Due to the large number of antennas and the high cost and power consumption of the RF chains at the mmWave frequency bands, the mmWave transceivers normally employ analog-only or hybrid analog digital architectures \cite{HeathJr2016,Li2019}. Following that, the mmWave transceiver is assumed to adopt an analog-only architecture with one RF chain and $M_\text{mmW}$ phase shifters. If $\bff \in \mathbb{C}^{M_\text{mmW} \times 1}$ denotes the downlink beamforming vector, then the received signal  at the mobile user can then be expressed as 
\begin{equation}
y_\text{mmW} [\bar{k}]=\bh_\mathrm{mmW}^T[\bar{k}]\bff s_\text{d} + n_\mathrm{mmW}[\bar{k}],
\end{equation}
where $\bh_\mathrm{mmW} [\bar{k}] \in \mathbb{C}^{M_\text{mmW} \times 1}$ represents the uplink channel from the mobile user to the BS mmWave array at the $\bar{k}$th subcarrier, $\bar{k}=1,2,...,\bar{K}$. Due to the hardware constraints on the mmWave analog beamforming vectors, these vectors are normally selected from quantized codebooks. Therefore, we assume that the beamforming vector $\bff$ can take one of candidate values collected in the codebook $\mathcal{F}$, i.e., $\bff \in \mathcal{F}$, with cardinality $\left|\mathcal{F}\right|=N_\mathrm{CB}$.  

\textbf{Channel Model:} This paper adopts a geometric (physical) channel model  for the sub-6GHz and mmWave channels \cite{HeathJr2016}. With this model, the mmWave channel (and similarly the sub-6GHz channel) can be written as 
\begin{equation}
\bh_\mathrm{mmW} [k] = \sum_{d=0}^{D-1} \sum_{\ell=1}^L \alpha_\ell e^{- j \frac{2 \pi k}{K} d} p\left(dT_\mathrm{S} - \tau_\ell\right) \ba\left(\theta_\ell, \phi_\ell\right),
\end{equation} 
where $L$ is number of channel paths, $\alpha_\ell, \tau_\ell, \theta_\ell, \phi_\ell$ are the path gains (including the path-loss), the delay, the azimuth angle of arrival (AoA), and elevation AoA, respectively, of the $\ell$th channel path. $T_\mathrm{S}$ represents the sampling time while $D$ denotes the cyclic prefix length (assuming that the maximum delay is less than $D T_\mathrm{S}$). Note that the advantage of the physical channel model is its ability to capture the physical characteristics  of the signal propagation including the dependence on the environment geometry, materials, frequency band, etc., which is crucial for our machine learning based beam and blockage prediction approaches. The parameters of the geometric channel models, such as the angles of arrival and path gains, will be obtained using accurate 3D ray-tracing simulations, as will be discussed in detail in \sref{sec:Results}.

\section{Problem Definition}

Adopting the dual-band system model described in \sref{sec:System}, the objective of this paper is to leverage the uplink channel knowledge at sub-6GHz band to enhance the achievable rate and reliability of the downlink mmWave link. More specifically, we focus on two important problems: (i) how can the  uplink sub-6GHz channel be exploited to find the optimal downlink mmWave beamforming vector that maximizes the achievable rate and (ii) how can the knowledge of the uplink sub-6GHz channel be used to infer whether or not the line-of-sight link to the mobile user is blocked. Next, we formulate these two problems. 

 \textbf{Problem 1: Beam Prediction} Consider the system and channel models in \sref{sec:System}, the downlink  achievable rate for a mmWave channel $\bh_\text{mmW}$ and a beamforming vector $\bff$ is written as 
\begin{equation} 
R\left(\left\{\bh_\text{mmW}[\bar{k}]\right\}, \bff \right)=\sum_{\bar{k}=1}^{\bar{K}} \log_2\left(1+\mathsf{SNR} \left|\bh_\text{mmW}[\bar{k}]^T \bff\right|^2\right),
\label{eq:ratef}
\end{equation}
with the per-subcarrier SNR defined as $\mathsf{SNR}=\frac{P_\text{mmW}}{K \sigma_\text{mmW}^2}$.  The optimal beamforming vector $\bff^\star$ that maximizes $R\left(\left\{\bh_\text{mmW}[\bar{k}]\right\}, \bff \right)$ is given by the exhaustive search
\begin{equation}  
\bff^\star=\argmax_{\bff \in \mathcal{F}} R\left(\left\{\bh_\text{mmW}[\bar{k}]\right\}, \bff \right), \label{eq:opt_f}
\end{equation} 
yielding the optimal rate $R^\star\left(\left\{\bh_\text{mmW}[\bar{k}]\right\} \right)$. 
For ease of exposition, we drop the sub-carrier indices in the rest of the paper; i.e., we will use $\bh_\text{mmW}$ and $\bh_\text{sub-6}$ to mean  $\{\bh_\text{mmW}[\bar{k}]\}$ and $\{\bh_\text{sub-6}[k]\}$.
It is important to note here that the beamforming vector $\bff$ is assumed to be implemented in the analog/RF domain as discussed in \sref{sec:System}. Therefore, the same beamforming vector is applied to all the subcarriers. Further, this beamforming vector can only be selected from the codebook $\mathcal{F}$. These constraints on the beamforming vector $\bff$ renders the achievable rate optimization problem as a non-convex problem with the optimal solution only found via the exhaustive search in  \eqref{eq:opt_f}.
Performing this  search, however, requires either estimating the mmWave channel $\bh_\text{mmW}$ or an online exhaustive  beam training, both of which are associated with large training overhead.  To reduce (or eliminate) this training overhead, the objective of this work is to exploit the sub-6GHz channels $\bh_\text{sub-6}$ to decide on the optimal beamforming vector. If $\hat{\bff} \in \mathcal{F}$ denotes the predicted beamforming vector based on the knowledge of $\bh_\text{sub-6}$, then the first objective of this work is to maximize the success probability in predicting optimal beamforming vector $\bff^\star$, defined as 
\begin{equation}
\kappa_1=\mathbb{P}\left(\hat{\bff}=\bff^\star\left|\left\{\bh_\text{sub-6}\right\}\right.\right).
\end{equation}

 \textbf{Problem 2: Blockage Prediction} The sensitivity of mmWave signals to blockages can critically impact the reliability of the high frequency systems.  If the status of the link in terms of line-of-sight (LOS) (unblocked) or non-LOS (blocked) can be predicted, this can enhance the system reliability via, for example, proactively handing over the user to another base station/access point \cite{Alkhateeb2018a}. In this work, we explore the possibility of using sub-6 GHz channels to predict whether the link connecting the base station and the user is blocked (NLOS) or unblocked (LOS).  let $s\in\mathcal{B}$ denote the correct (ground-truth) blocked/unblocked status of the communication link between the base station and the user, with $s = 1$ indicating a blocked link and $s=0$ indicating an unblocked link.  If $\hat{s}$ is the predicted link status using the sub-6 GHz channel knowledge, then the objective of the second problem in this work is to maximize the success probability of predicting the correct blockage status defined as 
 \begin{equation}
 \kappa_2=\mathbb{P}\left(\hat{s}=s\left|\left\{\bh_\text{sub-6}\right\}\right.\right).
 \end{equation}

In the next two sections, we present our proposed solutions that leverage machine learning tools to address the formulated mmWave beam and blockage prediction problems.

\section{Predicting mmWave Beams using Sub-6GHz Channels} \label{secIII}

Enabling the high data rate gains at mmWave communication systems requires the deployment of large antenna arrays at the transmitters and/or the receivers. Finding the best beamforming vectors $\bff^\star$ for these arrays is normally done through an exhaustive search over a large codebook of candidate beams, which is associated with large training overhead \cite{Wang2009,Hur2013}. In this section, we investigate the feasibility of exploiting sub-6GHz channels to predict/infer mmWave beams. If this is possible, we can expect dramatic savings in the mmWave beam training overhead as sub-6GHz channels can be easily estimated with a few pilots; ideally one pilot is required to estimate the uplink sub-6GHz channels. Leveraging sub-6GHz channels to predict mmWave beams is also motivated by the fact that  future wireless networks, such as 5G, will likely be dual-band---operating at both sub-6GHz and mmWave bands. Next, we will first reveal in \sref{subsec:Map_beam} that for any given environment, there exist a deterministic  mapping from sub-6GHz channels to the optimal mmWave beams under certain conditions. Then, we will show in \sref{subsec:DL_beam} how deep learning models can be exploited to predict the optimal mmWave beams using sub-6GHz channels with a probability of error that can be made arbitrarily small.  

\subsection{Mapping Sub-6GHz Channels to mmWave Beams} \label{subsec:Map_beam}

This section establishes the theoretical foundation for our proposed solution that predicts mmWave beams using sub-6GHz channels. More specifically, we will prove that, under certain condition, there exists a deterministic mapping from sub-6GHz channels to mmWave channels and beams. This proof extends the channel mapping concept that we  proposed in \cite{Alrabeiah2019}. First, consider the dual-band system and channel models described in \sref{sec:System}. Let $\{\bx_u\}$ represent the set of candidate user positions, with $\bx_u$ denoting the position of user $u$. Further, let $\bh^u_\text{sub-6}, \bh^u_\text{mmW}$ denote the channels from user $u$ to the sub-6GHz and mmWave antenna arrays. Now, we can define the following mapping functions, $\boldsymbol{\psi}_\text{sub-6},\boldsymbol{\psi}_\text{mmW}$, from the set of candidate positions $\{\bx_u\}$ to the corresponding sub-6GHz and mmWave channels,
\begin{align}
& \boldsymbol{\psi}_\text{sub-6}:\{\bx_u\} \rightarrow \{\bh^u_\text{sub-6}\}, \\ 
& \boldsymbol{\psi}_\text{mmW}:\{\bx_u\} \rightarrow \{\bh^u_\text{mmW}\}.
\end{align} 

Further, for any given  mmWave channel $\bh^u_\text{mmW}$ and beamforming vector $\bff_n \in \mathcal{F}$, the achievable rate $R\left(\bh^u_\text{mmW},\bff_n\right)$ is calculated by \eqref{eq:ratef}. Based on that, we define the position to achievable rate mapping functions $\bg^n(.), n=1,2,..., \left|\mathcal{F}\right|$ as 
\begin{equation}
\bg^n: \{\bx_u\} \rightarrow \{R\left(\bh^u_\text{mmW},\bff_n\right) \}, \ \ n=1, 2, ..., \left|\mathcal{F}\right|.
\end{equation}
\noindent Note that the existence of these mapping functions $\bg^n, \forall n$ follows directly from the existence of the position to mmWave channel mapping, $\boldsymbol{\psi}_\text{mmW}$, and the deterministic achievable rate function in \ref{eq:ratef} that relates the mmWave channels and the achievable rates with the $\left|\mathcal{F}\right|$ beamforming vectors. Next, we follow \cite{Alrabeiah2019} and adopt the following bijectiveness assumption of the mapping function $ \boldsymbol{\psi}_\text{sub-6}$ that maps the positions to sub-6GHz channels.
\begin{assumption} \label{assume1}
	The position to sub-6GHz channel mapping function, $\boldsymbol{\psi}_\emph{\text{sub-6}}$, is bijective.
\end{assumption}

Assumption \ref{assume1} means that any two user positions in $\{\bx_u\}$  have different sub-6GHz channel vectors, i.e., two  positions can result in the same sub-6GHz channels. This bijectiveness condition depends on the number of antennas, the array geometry, the number of paths, and the surrounding environment among other factors. In general, however, it is possible to show that a few antennas could be sufficient to make this bijectiveness assumption satisfied with very high probability \cite{Alrabeiah2019,Vieira2017}. The importance of this bijectiveness condition in Assumption \ref{assume1} is that it guarantees the existence of the inverse mapping $\boldsymbol{\psi}^{-1}_\emph{\text{sub-6}}(.)$ that maps the sub-6GHz channels in $\{\bh^u_\text{sub-6}\}$ to the corresponding positions in $\{\bx_u\}$. Next, we present the main proposition on the existence of the mapping from the sub-6GHz channels to the optimal mmWave beams.
\begin{proposition} For any given communication environment, and under Assumption \ref{assume1}, there exists a sequence of sub-6GHz to achievable rate mapping functions $\boldsymbol{\Phi}^n_\emph{\text{sub-6}}, n=1, 2, ..., \left|\mathcal{F}\right|$ that equal 
	\begin{equation}
			\boldsymbol{\Phi}^n_\emph{\text{sub-6}}: \{\bh_\emph{\text{sub-6}}^u\} \rightarrow \{R\left(\bh^u_\emph{\text{mmW}},\bff_n\right)\}, \ \ n=1, 2, ..., \left|\mathcal{F}\right|,
	\end{equation}
and with the optimal mmWave beamforming vector $\bff^\star$ for user $u$ obtained as $\bff^\star=\bff_{n^\star}$, with
\begin{equation}
n^\star=\argmax_{n \in \{1,2,...,\left|\mathcal{F}\right|\}} \ \boldsymbol{\Phi}^n_\emph{\text{sub-6}}\left(\bh_\emph{\text{sub-6}}^u\right).
\end{equation}
\label{prop1}
\end{proposition}
\begin{proof}
	The proof follows from the existence of the sub-6GHz channel to position mapping function $\boldsymbol{\psi}^{-1}_\text{sub-6}(.)$ and the existence of the position to mmWave achievable rate mapping functions $\bg^n(.)$. This leads to the existence of the composite mapping functions $\boldsymbol{\Phi}^n_\text{sub-6}$ since the co-domain of $\boldsymbol{\psi}^{-1}_\text{sub-6}(.)$ is the same as the domain of $\bg^n(.)$, and both equal to $\{\bx_u\}$. Finally, since the mapping functions $\boldsymbol{\Phi}^{n}_\text{sub-6}(.)$  result in the achievable rates with the candidate beams, the optimal beamforming vector $\bff^\star$ is found via the exhaustive search in \eqref{eq:opt_f}.
\end{proof}

Proposition \ref{prop1} shows that, under certain conditions, there exist mapping functions $\boldsymbol{\Phi}^n_\text{sub-6}, \forall n$, that can be leveraged to predict the optimal mmWave beam using sub-6GHz channels. Despite the existence of this mapping, though, it is very hard to leverage it using classical (non machine learning) solutions as this mapping functions are normally very hard to be characterized analytically. This motivates utilizing deep learning to learn these non-trivial mapping functions.  

\subsection{Deep Learning Based Beam Prediction} \label{subsec:DL_beam} 

Deep learning models have the interesting capability of learning and approximating non-trivial functions. Leveraging these models can effectively enable the prediction of the optimal mmWave beams directly from the knowledge of the sub-6GHz channels with an arbitrarily small error. Next, we use the universal approximation theory \cite{UnivApprox}, to prove that. 
\begin{proposition}	\label{prop2}
Let $\boldsymbol{\Pi}_{N}^n(.)$ represent the output of a dense neural network that consists of a single hidden layer of $N$ nodes. Then, for any $\epsilon_n>0$, and a continuous achievable rate mapping function $\boldsymbol{\Phi}_\emph{\text{sub-6}}^n\left(.\right)$, there exists a positive constant $N$ such as 
\begin{equation}
\sup_{\bh \in \{\bh_{\emph{\text{sub-6}}^u}\}} \left|\boldsymbol{\Pi}_{N}(\bh, \Omega)-\boldsymbol{\Phi}_\emph{\text{sub-6}}^n\left(\bh\right)\right|<\epsilon_n,
\end{equation}
where $\Omega$ denotes the set of neural network parameters.
\end{proposition}
\begin{proof}
	Proposition \ref{prop2} follows directly from the universal approximation theorem \cite[Theorem 2.2]{Hornik1989} by noticing that the set of sub-6GHz channels $\{\bh_\mathrm{sub-6}\}$ is a compact set since it is closed and bounded.
\end{proof}

Since the function $\boldsymbol{\Phi}_\text{sub-6}^n\left(\bh\right)$ maps the sub-6GHz channels to the mmWave achievable rate using the beamforming vector $\bff_n \in \mathcal{F}$, Proposition \ref{prop2} simply means that using a large enough neural network, we can predict the mmWave achievable rate $\hat{R}(\bh_\text{mmW},\bff_n)$ associated with every beam $\bff_n \in \mathcal{F}$ with arbitrarily small error. Next, we make an assumption on the codebook $\mathcal{F}$ before presenting the main result in Corollary \ref{cor1}.

\begin{assumption} \label{assume_CB}
	The mmWave beamforming codebook $\mathcal{F}$ satisfies the following condition
	\begin{equation}
 R\left(\bh_\emph{\text{mmW}}, \bff^\star\right) - R\left(\bh_\emph{\text{mmW}}, \bff_n\right) > 0, \ \ 	\forall \bh_\emph{\text{mmW}} \in \{\bh^u_\emph{\text{mmW}}\},
	\end{equation}
	where $\bff_n, \bff^\star \in \mathcal{F}$ and $\bff_n \neq \bff^\star$.
\end{assumption}

Assumption 1 simply means that there is only one optimal beamforming codeword for any channel $\bh_\text{mmW} \in \left\{\bh^u_\text{mmW}\right\}$. It is important to note here that while we need this assumption to prove the result in the following Corollary, violating this condition on the codebook $\mathcal{F}$ leads to the trivial case where two beamforming vectors can achieve exactly the optimal rate. Next, we present Corollary \ref{cor1} that establishes the feasibility of predicting the optimal mmWave beams using sub-6GHz channels via deep neural networks.  

\begin{corollary} \label{cor1}
Let $\boldsymbol{\Pi}^n_{N}(.), n=1,2,...,|\mathcal{F}|$ represent the output of a dense neural network that consists of a single hidden layer of $N$ nodes. Further, define the predicted beamforming vector $\hat{\bff}=\bff_{\hat{n}} \in \mathcal{F}$ with $\hat{n}=\argmax_{n=1,2,...,|\mathcal{F}|} \boldsymbol{\Pi}^n_{N}(.)$. Then, for any $\epsilon>0$, and continuous achievable rate mapping functions $\boldsymbol{\Phi}_\emph{\text{sub-6}}^n\left(.\right), n=1,2,...,|\mathcal{F}|$, there exists s positive constant $N$ large enough such as 
\begin{equation}
\kappa_1=\mathbb{P}\left(\hat{\bff}=\bff^\star\left|\bh_\emph{\text{sub-6}}\right.\right) > 1- \epsilon.
\end{equation}
\end{corollary}
\begin{proof}
	See Appendix \ref{app1}.
\end{proof}

Corollary \ref{cor1} is very interesting as it proves that it is possible to use deep neural networks to predict the \textit{optimal} mmWave beamforming vector directly from the knowledge of the sub-6GHz channels once the achievable rate mapping functions, $\boldsymbol{\Phi}_\text{sub-6}^n\left(.\right), n=1,2,...,|\mathcal{F}|$, exist. Further, we know from Proposition \ref{prop1} that the existence of these mapping functions for any given environment requires only that the mapping from the candidate set of positions to the sub-6GHz channels is bijective -- a condition that is achievable with high probability as we discussed earlier in \sref{subsec:Map_beam}. 

\textbf{Proposed Deep Learning Based System Operation:} Consider the dual-band system model in \sref{sec:System}. The proposed system operation that exploits deep learning to predict the optimal mmWave beam directly from the sub-6GHz channels operates in the following two phases: 
\begin{itemize}
\item \textbf{Deep Learning Training Phase:} In this phase, the dual-band communication system operates as if there is no machine learning: For every coherence time, the uplink sub-6GHz channel is estimated requiring only one uplink pilot, and an exhaustive beam training is done for the mmWave downlink to calculate the achievable rate using every beamforming vector.  Let  $\mathcal{R}^u=\left\{R(\bh_\text{mmW}^u, \bff_1), ...., R(\bh_\text{mmW}^u, \bff_{\left|\mathcal{F}\right|})\right\}$ denote the set of achievable rates at user $u$ for all the codebook beams. Then, at every coherence time, one new data point $\left(\bh_\text{sub-6}^u,\mathcal{R}^u\right)$ is added to the deep learning dataset. After collecting large number of data points, we use this dataset to train the deep learning model, which will be described in detail in \sref{sec:DL_model}. 
\item \textbf{Deep Learning Deployment Phase:} Once the deep learning model is trained, the base station uses it to directly predict the optimal mmWave beam using the sub-6GHz channels. More specifically, this phase requires the user to send only one uplink pilot to estimate the sub-6GHz channels and this channel is passed to the deep learning model which predicts which mmWave beam should be used for the downlink mmWave data transmission. This saves all the training overhead associated with the mmWave exhaustive beam training process.
\end{itemize} 

It is important to note here that \textbf{the proposed deep learning based system operation has almost no learning overhead in terms of the system time-frequency resources}. That is because the mmWave beam training will typically be performed anyway, in systems that do not use machine learning, to figure out the best beamforming direction. This means that the dataset collection process and the deep learning training are done without affecting the classical mmWave system operation. Hence, even if a large dataset needs to be collected to capture the dynamics in the environment, that is feasible because it does not interfere with the classical system operation.

\textbf{Practical Challenges:}
As shown in this section, for any given static environment, once the mapping from the candidate positions to the sub-6GHz channels is bijective (one-to-one), the sub-6GHz channels can be exploited to directly predict the optimal mmWave beams with a very high success probability. In practice, however, there are a few factors that can add some probabilistic error to this beam prediction such as the measurement noise, the phase noise, and the dynamic scatterers in the environment. These factors can make the position-to-channel mapping not perfectly bijective or create sub-6GHz channels that are different than those experienced before by the neural networks.  In \sref{sec:Results}, we will evaluate the impact of some of these practical considerations on the beam prediction performance.

\section{Predicting mmWave Blockages using Sub-6GHz Channels} \label{sec:blockage}

The reliability of the communication links is one of the main challenges for mmWave systems. This is mainly because of the sensitivity of mmWave signals to blockages, which can result in a sudden drop in the SNR if the line-of-sight (LOS) path is obstructed. With this motivation, \cite{Alkhateeb2018a} proposed to leverage machine learning tools to learn the mobility patterns of the transmitters, receiver, or scatterers, and hence predict  blockages before they actually block the LOS path. This can enable the network to act proactively, for example by handing over the communication session to another base station, before the session is disconnected. In this paper, we focus on a different but equally important problem which is the ability of the dual-band base stations to use the sub-6GHz channels to decide whether or not the mmWave LOS link is blocked. This knowledge can potentially help the BS in adapting its transmission strategy accordingly by, for example, changing the transmit power and modulation/coding scheme or handing off the communication session to the sub-6GHz band. In \sref{subsec:map_blockage}, we will investigate the conditions under which the sub-6GHz can indicate the LOS blockage/no-blockage status. Then, we show in \sref{subsec:DL_blockage} that this capability can be implemented using deep neural networks.

\subsection{Mapping Sub-6GHz Channels to Link Blockages} \label{subsec:map_blockage}

Consider the system model in \sref{sec:System} with co-located sub-6GHz and mmWave arrays at the base station. Let $\mathcal{X}=\{x_u\}$ represent the set of candidate user locations. To simplify the analysis in this section, we make the following assumption
\begin{assumption}
	For all the users in $\mathcal{X}$, if a blockage obstructs the LOS path to the mmWave array, it also obstructs the LOS path to the sub-6GHz array.
\end{assumption}
Note that this assumption is typically satisfied in practice since the sub-6GHz and mmWave arrays are co-located. It is also worth mentioning here that while obstructing the mmWave LOS link may completely block the link (due to the high penetration loss at mmWave), the obstruction of the sub-6GHz LOS ray will likely only reduce its power without a complete blockage. Our analysis, however, is general and independent of whether the LOS rays are completely of partially blocked.  Now, define $s_u \in \mathcal{S}=\{0, 1\}$ as the blockage status of user $u$, with $s_u=0$ and $s_u=1$ indicating that the LOS path between user $u$ and the BS is, respectively, unblocked or blocked. For a given environment, let $\bh_\rm{sub-6, B}^{u}$ denote the sub-6GHz channel of user $u$ when the LOS path is obstructed/blocked and $\bh_\rm{sub-6, UB}^{u}$ denote the channel when the LOS path is not blocked. Further, let  $\mathcal{H}_\mathrm{B}=\{\bh_\rm{sub-6, B}^{u}\}$ and $\mathcal{H}_\mathrm{UB}=\{\bh_\rm{sub-6, UB}^{u}\}$ represent the blocked and unblocked channel sets. Next, we define the mapping function $\Psi$ that maps the user position and blockage status to a sub-6GHz channel. 
\begin{equation}
\boldsymbol{\Psi}: \mathcal{X} \times \mathcal{S} \rightarrow \mathcal{H},
\end{equation} 
where $\mathcal{X} \times \mathcal{S}$ is the Cartesian product of the user position and blockage status sets, and $\mathcal{H}$ represent the set of all blocked and unblocked channels, i.e., $\mathcal{H}=\mathcal{H}_\mathrm{B} \cup  \mathcal{H}_\mathrm{UB}$. In the following proposition, we state the condition under which the LOS blockage can be identified using the sub-6GHz channels. 
\begin{proposition} \label{prop_blockage}
For any given environment, if the mapping function $\boldsymbol{\Psi}$ is bijective, then there exists a continuous discriminant function $f: \mathcal{H} \rightarrow {0,1}$ such that 
\begin{equation}
\forall \bh \in \mathcal{H}_\rm{B}, f(\bh)=1, \ \   \ \ \forall \bh \in \mathcal{H}_\rm{UB}, f(\bh)=0.
\end{equation} 
\end{proposition}
\begin{proof}
	When the mapping $\boldsymbol{\Psi}$ is bijective, each $(\bx_u, s_u)$ tuple has a unique channel, which yields disjoint blocked and unblocked channel sets, i.e., $\mathcal{H}_\rm{B} \cap \mathcal{H}_\rm{UB}= \phi$. This leads to the existence of the continuous discriminant function $f(.)$ using the Urysohn Lemma \cite{Kotze1992}.
\end{proof}

Note that the bijectiveness condition of the mapping function $\boldsymbol{\Psi}$ means that (i) every user position in $\{\bx_u\}$ will yield two different channels for the LOS-obstructed or unobstructed cases and that (ii) these  LOS-obstructed/unobstructed channels are different for all the users in $\{\bx_u\}$ . Similar to Assumption \ref{assume1}, the bijectiveness condition of $\boldsymbol{\Psi}$ is expected to be satisfied with high probability in multi-antenna systems, as will be shown in \sref{subsec:block_results}.

\subsection{Deep Learning Based Blockage Prediction} \label{subsec:DL_blockage}

Using sub-6GHz channel knowledge, Proposition \ref{prop_blockage} proves that it is possible to decide whether the LOS link between the base station and user is blocked or not under some conditions. The discriminating function that does this, however, is hard to be characterized analytically and may be highly non-linear given the nature of the complex channel vectors. Intuitively, deciding whether the LOS link is blocked or not requires some spatial and power analysis of the rays that construct the channels which is a non-trivial task. Motivated by these challenges, we propose to leverage the powerful learning capabilities of deep neural networks to learn this LOS  blockage discriminating function. This is addressed by the following proposition.

\begin{proposition} \label{prop:block_DL}
		Let $\boldsymbol{\Pi}^n_{N}(.), n=1,2,...,|\mathcal{F}|$ represent the output of a dense neural network that consists of a single hidden layer of $N$ nodes. Further, define the predicted blockage status of using this network as $\hat{s}_{N}$. If the conditions in Proposition \ref{prop_blockage} are satisfied, then for any $\epsilon>0$, there exists s positive constant $N$ large enough such as 
		\begin{equation}
		\kappa_2=\mathbb{P}\left(\hat{s}_N=s\left|\bh_\emph{\text{sub-6}}\right.\right) > 1- \epsilon.
		\end{equation}
\end{proposition}	
\begin{proof}
	The proof is similar to that in Corollary \ref{cor1} and is omitted due to space limitation.  
\end{proof}

\textbf{Practical Challenges:} Proposition \ref{prop:block_DL} highlights the interesting ability of neural networks in classifying the sub-6GHz channel to LOS blocked or unblocked classes. One important challenge in this application, however, is obtaining the ground-truth blocked/unblocked labels. Therefore, it is important to develop practical labeling techniques that construct the required labels for training the neural networks. In \sref{subsec:block_results}, we propose a labeling strategy based on analyzing the mmWave beam training results and evaluate its performance compared to the case when the ground-truth labels are available.

\section{Deep Learning Model} \label{sec:DL_model}

In Sections \ref{subsec:DL_beam}  and \ref{subsec:DL_blockage}, we proved theoretically how neural networks can enable the prediction of the mmWave beams and blockages using sub-6GHz channels. In this section, we will describe our specific design of the neural network architecture and the adopted learning model. Before we delve into the description of the proposed model, it is important to note that the two tasks we consider in this paper, namely predicting the optimal mmWave beam and predicting the link blockage status, involve a selection from a pre-defined set of options--a beam codebook or a binary set of blocked/unblocked status. These problems have then a striking similarity with the well-known classification problem in machine learning \cite{PRML_Bishop}. Specifically, in the beam prediction problem, each sub-6GHz channel is mapped to one of $D=\left|\mathcal{F}\right|$ indices, where $D$ is the size of discrete set of options. This could be viewed as a classification problem where each beam index represents a class, and the job of the learning model is to learn how to \textit{classify} channels into beam indices. In the recent years, deep-architecture neural networks have performed exceedingly well in handling classification problems \cite{AlexNet}\cite{VGG}\cite{ResNet}, among other things. Motivated by these results and by the conclusions of Sections \ref{subsec:DL_beam}  and \ref{subsec:DL_blockage}, we design a deep neural network model to address the mmWave beam and blockage prediction problems.

\subsection{Deep Neural Network Design}  \label{subsec:DNN}
The first step in designing a neural network is the choice of the network type, which should be based on the nature of the problem and the desired role of the model. For our beam/blockage prediction problems, the objective is to learn how to map the sub-6GHz channel vectors to a real-valued $D$-dimensional vector $\bp$, where $D$ is either the codebook size, $\left|\mathcal{F}\right|$, or 2 for the blockage status. For this objective, and motivated by the universal approximation results in Sections \ref{subsec:DL_beam}  and \ref{subsec:DL_blockage}, we adopt a Multi-Layer Perceptron (MLP) network, which comprises a sequence of non-linear vector transformations \cite{NN_Haykin}. The proposed network architecture has two main sections, namely the base network and the task-specific layer. 

 \textbf{Base Network:} The beam and blockage predictions are both posed as classification problems and both share the same input data (Sub-6GHz channels). Therefore, to reduce the computational burden of the training process, we propose to have a common neural network architecture for the two problems, which, as will be shown shortly, enables leveraging \textit{transfer learning} to reduce the training overhead. Based on that, a single \textit{base} deep neural network is designed for the two prediction problems. This network comprises $L_\mathrm{NN}$ stacks of layers, each of which has a sequence of fully-connected with ReLU non-linearity and dropout layers, as illustrated in Figure \ref{Arch}. All fully-connected layers have the same breadth, $M_\mathrm{NN}$ neurons per layer. 
	
	\begin{figure*}[t]
		\includegraphics[width=\linewidth]{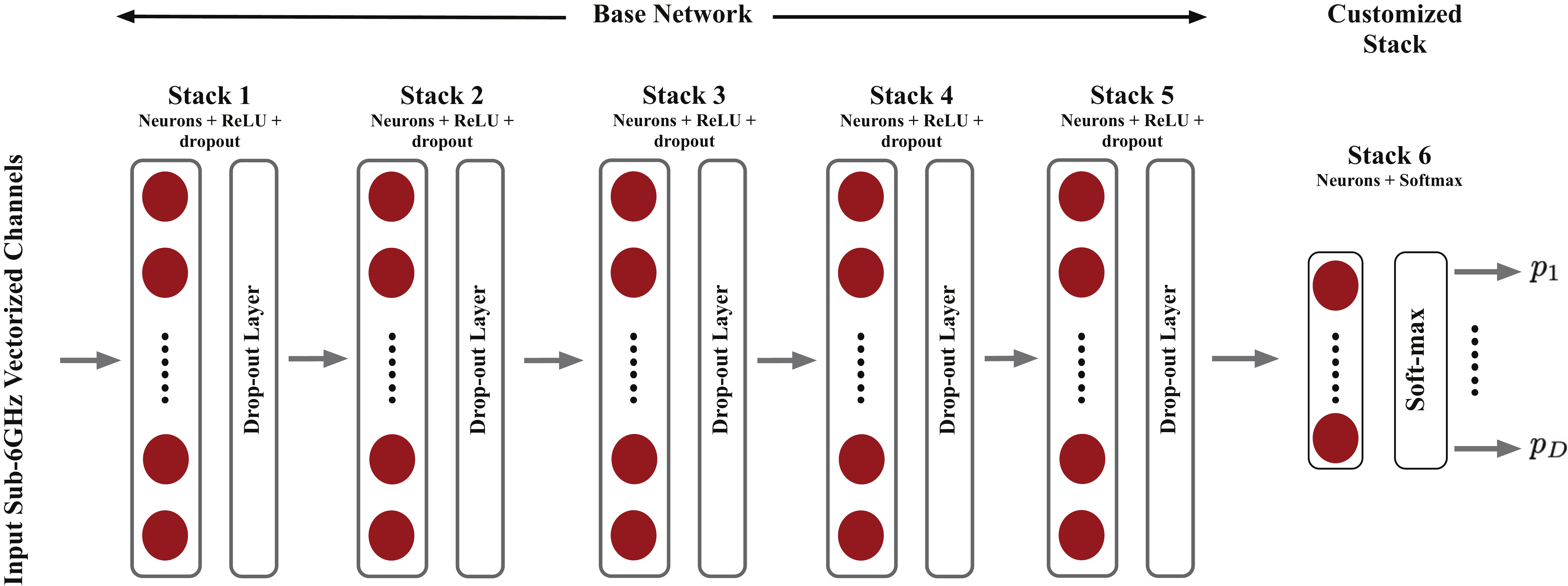}
		\caption{The overall deep neural network architecture. The first $L_{NN}$ stacks comprising multiple fully-connected, ReLU, and dropout layers form the base network. The final stack represents the customizable output layers. For both problems, it comprises a fully-connected layer followed by a soft-max. Their size depends of the number of classes in each task.}
		\label{Arch}
	\end{figure*}

\textbf{Task-Specific Output Layer:} The number of outputs in each prediction task (beam or blockage) differs as the number of classes changes; predicting a beam index means that there are $D = \left|\mathcal{F}\right|$ beam choices, while predicting blockage is a binary problem with $D = 2$ choices, \textit{blocked} or \textit{unblocked}. Hence, the base network is customized with an additional stack of layers that depends on the target task. For beam prediction, the final layer is designed to have a fully-connected layer with $D = \left|\mathcal{F}\right|$ neurons. It acts as a linear classifier that projects its $M_\mathrm{NN}$-dimensional input feature vector onto a $D$-dimensional classification space. The projection is fed to a \textit{Softmax} layer, which induces a probability distribution over all the available classes. Formally, it does so by computing the following formula for every element $d$ in its input vector:
	\begin{equation}
	p_d= \frac{e^{z_d}}{\sum_{i = 1}^{D} e^{z_i}},
	\end{equation}
	where $z_i, i=1, ..., D$ is the $i$th element of the $D$-dimensional projection vector (input to the softmax), and $p_d$ is the probability that the $d$th beamforming vector is the correct prediction--more on Softmax could be found in \cite{DLBook}. Finally, the index of the element with the highest probability is the index of the predicted beam-forming vector. For the blockage prediction task, a similar last stack is designed but with different dimensions. The classifier has $D = 2$ neurons, and the Softmax here produces two probabilities, namely blocked ($p_1$) and unblocked ($p_2$).
	
 \textbf{Transfer Learning:} An interesting and advantageous characteristic of deep neural networks is their ability to exploit a learned function on a certain input data to perform another function on the same input data, which is referred to as \textit{transfer learning}. In \cite{TransfLearn}, it has been empirically shown that layers closer to the input learn generic features, i.e., those layers tend to learn the same mapping regardless of the task and final outputs of the neural network. However, as layers get farther away from the input and deeper into the network, features become more specific, i.e., they are more groomed to the task in question. Such empirical evidence suggests that reusing a trained network for a different task could provide an interesting boost in the network performance and help reduce the computational complexity associated with its training \cite{TransfLearn}.  
 
Now, given that both beam/blockage prediction problems could be faced by the same mmWave system, a resourceful way for good prediction performance in both cases is to apply transfer learning. As it will be discussed in Section \ref{sec:Results}, beam prediction is a more challenging problem than blockage prediction. This is mainly, but not exclusively, due to the large number of classes beam prediction has. Hence, the proposed training strategy in this work focuses on first training and testing the deep neural network for beam prediction. Once that cycle is done, the last stack of the trained network is replaced with that suitable to blockage prediction. Then, it undergoes another training and testing cycle (called fine-tuning) for the blockage prediction task. This offers faster convergence and improved performance compared to training from scratch for the blockage prediction task.
	
\subsection{Learning Model} \label{subsec:learn}
Our objective is to leverage the neural network architecture described in \sref{subsec:DNN} to learn how to predict mmWave beams and blockages directly from the sub-6GHz channels. To achieve that, we adopt a supervised learning model that operates in two modes, a background training mode and a deployment mode. Next, we explain the two modes. 

\textbf{1. Background Training Mode:} As described earlier in \sref{subsec:DL_beam}, the dual-band system operates as if there is no deep learning. It collects data points for the beam prediction dataset, $\left(\bh_\text{sub-6}^u,\mathcal{R}^u\right)$, and, if the blockage status knowledge is available, it collects data points for the blockage prediction dataset, $\left(\bh_\text{sub-6}^u, s_u\right)$. We will discuss how to obtain the blockage labels shortly. Both datasets needs to undergo pre-processing before being used for model training \cite{Alrabeiah2019}:
\begin{itemize}
	\item \textbf{input normalization:} The sub-6GHz channels, which are the inputs to the neural network, are normalized using a global normalization factor. Let $\Delta=\max_{\forall u, \forall i, \forall k} \left|\left[\bh_\text{sub-6,k}^u\right]_{i}\right|$ denote the global normalization factor where $\left[\bh_\text{sub-6,k}^u\right]_i$ is the $i$th element in the sub-6GHz channel vector of the $k$th subcarrier of user $u$. Then the sub-6GHz channels in the dataset $\left\{\bh_\text{sub-6,k}^u\right\}$ are all normalized by $\Delta$ to have a maximum absolute value of $1$. Every normalized channel is decomposed into real and imaginary vectors that are stacked together to form a real-valued vector. Finally, all real-valued vectors of the $K$ sub-6GHz subcarriers of a user $u$ are stacked together to form a ($2\times K \times M_{\text{sub-6}}$)-dimensional vector, the input to the neural network. Writing the complex channel vector as a real-valued vector of the stacked real, imaginary, and subcarriers is to enable the implementation of real-valued neural networks. 
	
	\item \textbf{Labels construction:} The labels are modeled as $D$-dimensional one-hot vectors\footnote{One-hot vector refers to a binary vector where all elements are zero except for a single element with the value of one.} indicating the class labels. For the beam prediction dataset, the one-hot vector for every sub-6GHz channel has 1 at the element that corresponds to the index of the optimal beamforming vector (which is calculated from \eqref{eq:opt_f}). For the blockage prediction task, the one-hot vectors are 2-dimensional with $[1, 0]$ for blocked and $[0, 1]$ for unblocked. In \sref{subsec:block_results}, we study the learning performance in two situations: (i) when the ground-truth blockage status is available and (ii) when the blockage status is estimated based on the angular distribution of receive power. 	
\end{itemize}

After preparing the dataset, the neural network model is trained to minimize the cross-entropy loss function, $L_\text{cross}$ defined as
\begin{equation}
L_\text{cross} = - \sum_{d=1}^D t_d \log_2(p_d),
\end{equation}
where $\bt=[t_1, ..., t_D]$ is the target one-hot vector and $\bp=[p1, ..., p_D]$ is the network prediction. It is important to mention here that $p_d$ represents the neural network predicted probability that the input sub-6GHz channel belongs to the $d$th class. 

\textbf{2. Deployment Mode:} Once the neural network model is trained, it is then used to predict the mmWave beams and blockage status directly from the knowledge of the sub-6GHz channels. Please refer to \sref{subsec:DL_beam} for more details.

\section{Experimental Results} \label{sec:Results}
In this section, we evaluate the performance of the proposed mmWave beam and blockage prediction solutions using numerical simulations. First, we describe the adopted evaluation scenarios in \sref{subsec:scenarios}. Then, we explain the construction of the deep learning dataset and neural network training process in Sections \ref{subsec:dataset} and \ref{subsec:train}. Finally, we show and discuss the performance results of the sub-6GHz based mmWave beam and blockage prediction solutions in Sections \ref{subsec:PerfBeam} and \ref{subsec:block_results}.

\subsection{Evaluation Scenarios} \label{subsec:scenarios}
Two publicly available evaluation scenarios from the DeepMIMO dataset \cite{DeepMIMO} are considered in the simulations. These scenarios are constructed using the 3D ray-tracing software Wireless InSite \cite{Remcom}, which captures the channel dependence on the frequency. In the beam prediction evaluation, we consider the outdoor scenario 'O1' that is available at two frequencies: `O1\textunderscore28' at 28GHz and `O1\textunderscore3p5' at 3.5GHz. For this scenario, we adopt a single base station (BS 3) and equip it with two co-located uniform linear arrays (ULAs) at 28GHz and 3.5GHz. In the blockage prediction evaluation, we consider the blockage scenario, named `O1\textunderscore28B' at 28GHz and `O1\textunderscore283p5B' at 3.5GHz. This scenario is identical to the first scenario but with a 6-meter tall blockage right in front of the base station (BS 3) and two possible reflecting surfaces on both sides of the BS, as depicted in \figref{fig:blocked}. This emulates a scenario where a large track, for example, comes between the users and the BS, and a couple of parked cars provide secondary signal paths to the BS. Note here that not all users in this scenario are blocked; those in the marked area of Figure \ref{fig:blocked} are considered blocked, as will be discussed in \sref{subsec:block_results}. 

\begin{figure}[t]
	\centering
	\includegraphics[width=.55\linewidth]{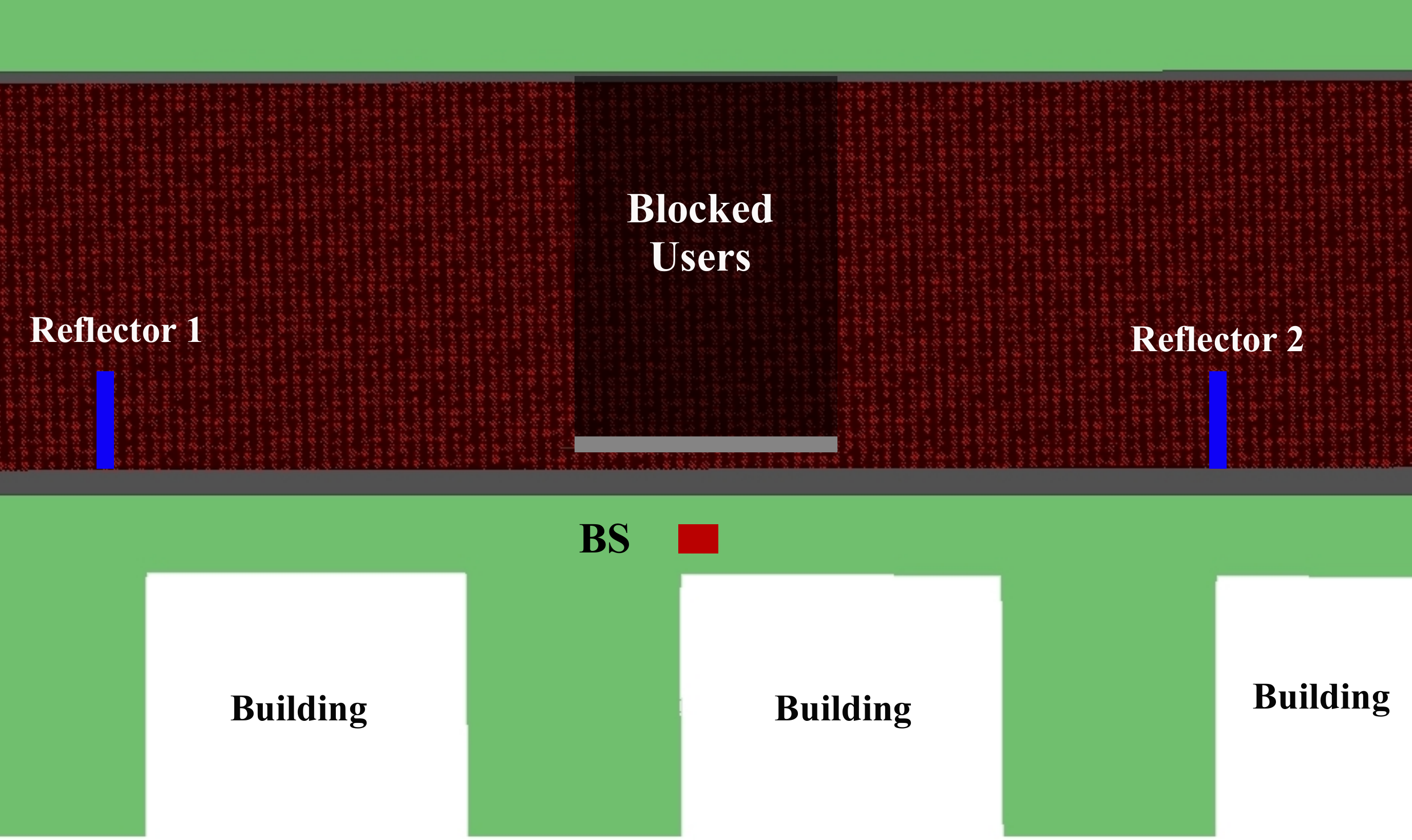}
	\caption{A top-view of the second scenario. It shows the BS location, the blockage, and two possible reflectors. The shaded area above the blockage marks the blocked users, and they are the only ones considered to construct the dataset.}
	\label{fig:blocked}
\end{figure}

\subsection{Dataset Generation} \label{subsec:dataset}

Given the two ray-tracing scenarios described in \sref{subsec:scenarios}, we construct the following two datasets for the beam and blockage prediction problems.
\begin{itemize}
	\item \textbf{Beam prediction dataset:} Here, we adopt the LOS scenario (`O1\textunderscore28' and `O1\textunderscore3p5') and use the DeepMIMO generator script \cite{DeepMIMO} with the parameters described in Table \ref{table:ScenPara}. This DeepMIMO script generates the sub-6GHz and mmWave channel sets $\{\bh^u_\text{sub-6}\}, \{\bh^u_\text{mmW}\}$, between the base station and every user $u$ in the scenario. Given these channels we construct the beam prediction dataset explained in \sref{subsec:learn}. Essentially, every data point in this dataset has the sub-6GHz channel and the corresponding one-hot vector that indicates the index of the optimal mmWave beam in the codebook $\mathcal{F}$. It is important to mention here that we adopt a simple quantized beam steering codebook where the $n$th beam$, n=1, 2, ..., \left|\mathcal{F}\right|$ is defined as $\bff_n=\ba(\frac{2 \pi n}{\left|\mathcal{F}\right|})$, with $\ba(.)$ representing the mmWave array response vector.     
	
	\item \textbf{Blockage prediction dataset:} This dataset considers the blockage scenario (`O1\textunderscore28B' and `O1\textunderscore3p5B') along with the LOS scenario (`O1\textunderscore28' and `O1\textunderscore3p5') and use the DeepMIMO generator script with the parameters in Table \ref{table:ScenPara}. The DeepMIMO script will generate the blocked sub-6GHz and mmWave channel sets with which half of the blockage dataset is constructed as described in \sref{subsec:learn} adopting only the users falling in the marked region of Figure \ref{fig:blocked}. The other half of the dataset is obtained from the LOS scenario; the users in the same marked region of Figure \ref{fig:blocked} but in the LOS scenario are selected for the blockage dataset. Each data point in the dataset consists of the sub-6GHz channel and the corresponding one-hot vector that indicates whether the LOS ray is obstructed (blocked) or not. 	
\end{itemize} 

\begin{table*}[h]
	\caption{DeepMIMO Dataset Parameters}
	\centering
	\begin{tabular}{c | c c | c c}
		\hline\hline
		\textbf{Parameters} & \multicolumn{2}{c}{\textbf{LOS}} & \multicolumn{2}{|c}{\textbf{Blockage}} \\
		\hline
		Scenario name & O1\textunderscore28 & O1\textunderscore3p5 & O1\textunderscore28B & O1\textunderscore3p5B \\
		Active BS & 3 & 3 & 3 & 3 \\
		Active users & 700-1300 & 700-1300 & 700-1300 & 700-1300 \\
		Number of BS Antennas & 64 & 4 & 64 & 4 \\
		Antenna spacing (wave-length) & 0.5 & 0.5 & 0.5 & 0.5 \\
		Bandwidth (GHz) & 0.5 & 0.02 & 0.5 & 0.02 \\
		Number of OFDM subcarriers & 512 & 32 & 512 & 32 \\
		OFDM sampling factor & 1 & 1 & 1 & 1 \\
		OFDM limit & 32 & 32 & 32 & 32 \\
		Number of paths & 5 & 15 & 5 & 15 \\
		\hline\hline
	\end{tabular}
	\label{table:ScenPara}	
\end{table*}

\subsection{Neural Network Training} \label{subsec:train}
In this paper, we adopt the deep neural network architecture, described in \sref{subsec:DNN}, with $L_\text{NN}=5$ stacks of layers and $M_\text{NN}=2048$ neurons per layer. This neural network is trained using the datasets, explained in \sref{subsec:dataset}, for the beam and blockage prediction tasks. The training, as well as testing, samples are first contaminated with noise depending on the target SNR. Then, the network is trained in one of two ways, from scratch or transfer learning. The training approach is different in the beam and blockage predictions tasks: (i) For the beam prediction problem, the neural network training follows the training from scratch approach, where the weights are initialized randomly. It uses an initial learning rate of $1\times10^{-1}$, which is dropped by a factor of $0.1$ after the $90^{th}$ epoch.  The other hyper-parameters are summarized in Table.\ref{TrainHyper}. (ii) For the blockage prediction problem, the  neural network is trained with transfer learning. The weights of the best-performing network trained for beam prediction are used to initialize those of the base model used for blockage prediction. The only part that is trained from scratch is the end-stack. This approach generally provides faster training convergence compared to the case wight random initial weights. All training and experiments are done in MATLAB using its Deep Learning toolbox running on a machine with an RTX 2080 Ti GPU. The code files are available online at \cite{main_code}.

\begin{table}[h]
	\caption{DNN training hyper-parameters}
	\centering
	\begin{tabular}{c | c | c}
		\hline\hline
		Parameter & Beam Prediction & Blockage Prediction \\
		\hline
		Solver & SGDM & SGDM \\
		\hline
		Learning rate & $1\times10^{-1}$ & $1\times10^{-1}$ \\
		\hline
		Momentum & 0.9 & 0.9 \\
		\hline
		Dropout percentage & 40\% & 40\% \\
		\hline
		$l_2$ Regularization & $1\times10^{-4}$ & $1\times10^{-4}$ \\
		\hline
		Max. number of epochs & 100 & 50 \\
		\hline
		Dataset size $|S_{band}|$ & $\approx 108\times10^3$ & $ \approx 47\times 10^3$ \\
		\hline
		Dataset split & 70\%-30\% & 70\%-30\% \\
		\hline\hline
	\end{tabular}
	\label{TrainHyper}	
\end{table}

\subsection{Performance Evaluation Metrics}
Given that the addressed beam/blockage prediction problems in this paper are formulated  as classification problems, we  adopt the \textit{Top-1} and \textit{Top-n} classification accuracies as the main performance metrics. The Top-1 accuracy, denoted $A_\text{Top-1}$,  is defined as the frequency at which the deep neural network correctly predicts the class of the input. Formally, it is written as
\begin{equation}\label{AccTop1}
A_\text{Top-1} = \frac{1}{N_\text{test}} \sum_{n=1}^{N_\text{test}} \mathds{1}_{\hat{d_n}=d_n^\star},
\end{equation}
where $\mathds{1}_{(.)}$ is the indicator function, and $\hat{d_n}, d_n^\star$  are the predicted and target classes of the $n$th test point. Further, owing to the fact that a classifying deep neural network produces a probability distribution over all  possible classes, it is interesting to study whether one of the top \textit{n} predictions is the correct class instead of only focusing on the Top-1 prediction. This is customarily quantified using the Top-n accuracy. It is defined as the frequency at which the neural network correctly predicts the class of the input within its top-n predictions. In terms of beam prediction, it means that we test whether the optimal mmWave beam is within the best $n$ predicted beam using the sub-6GHz channel. In addition to the Top-1 and Top-n accuracies, we also evaluate the performance of the proposed deep learning based model in terms of the achievable rates using the predicted mmWave beams.

\begin{figure}[t]
	\centering
	\subfigure[]{\includegraphics[width=.475\linewidth]{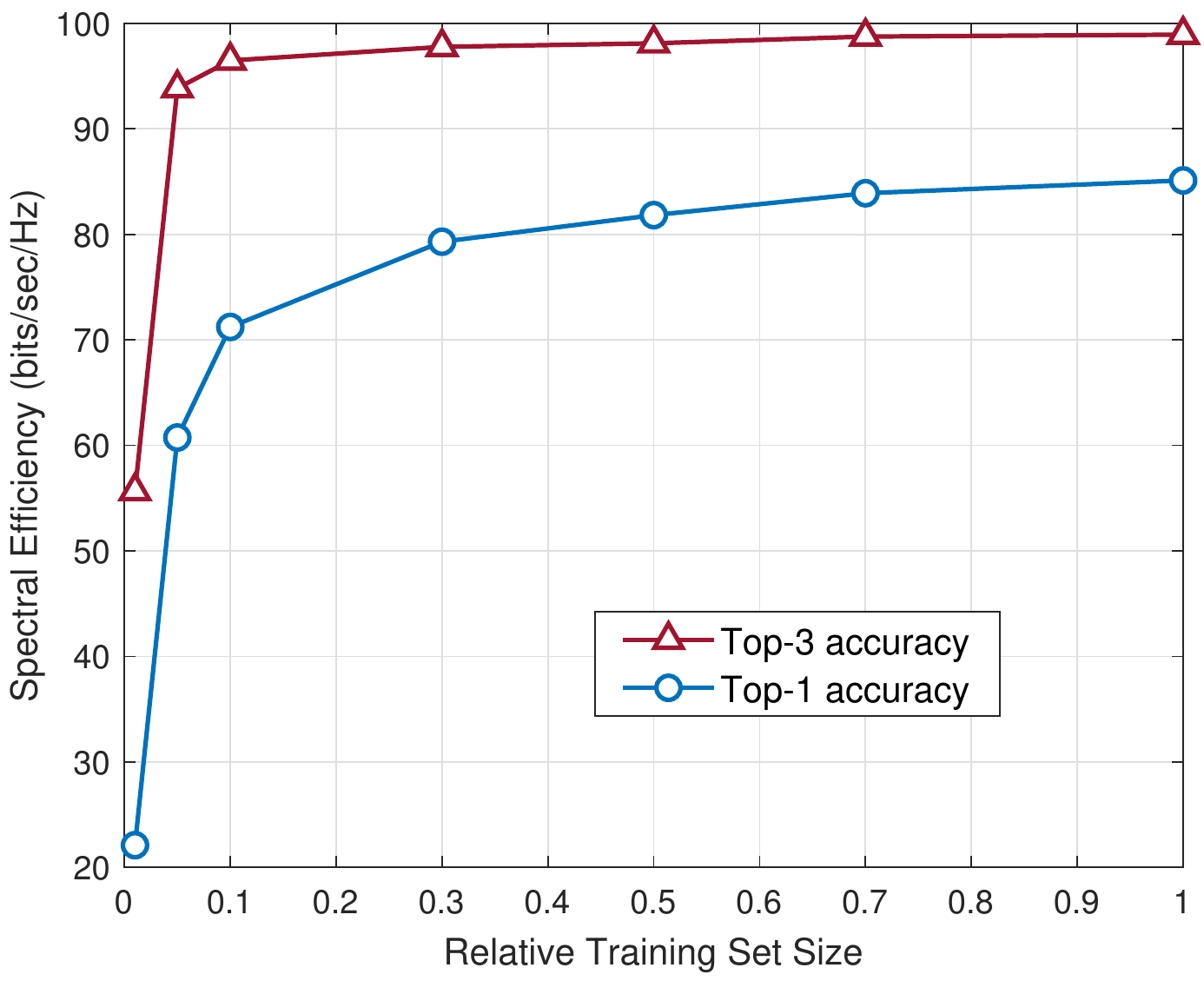}\label{fig:training}}
	\subfigure[]{ \includegraphics[width=0.475\linewidth]{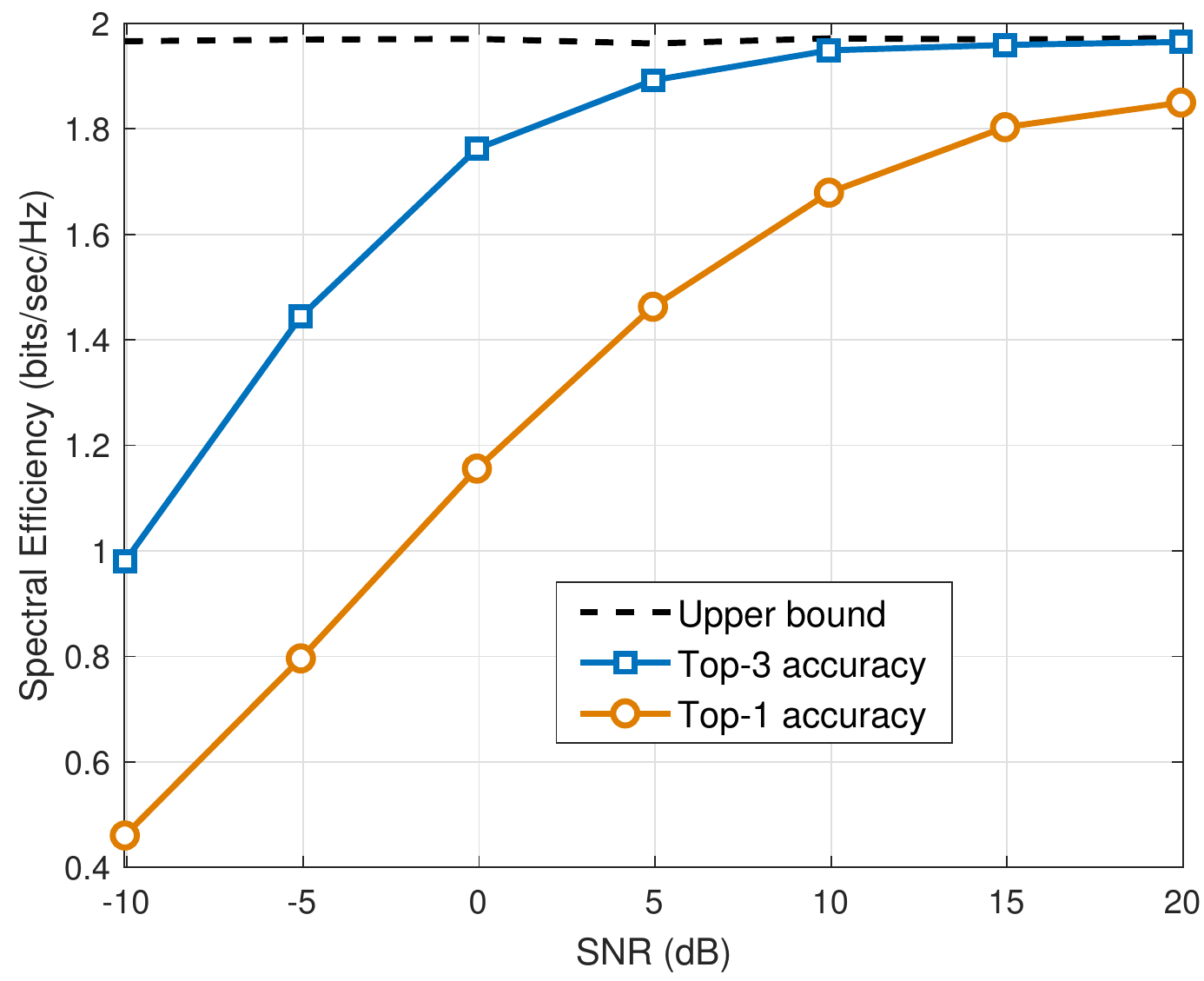}\label{perf_noise}}
	\caption{The effect of increasing the training set size on the beam-prediction performance is shown in (a), quantified by the top-1 and top-3 prediction accuracies. The values on the x-axis are relative to the total training set size, $\approx 76,000 $ data pairs. In (b), The performance of the deep learning solution is plotted when the sub-6GHz channels are contaminated with noise. The SNR represents the sub-6GHz receiver SNR.}
	
\end{figure}

\subsection{Beam Prediction Performance} \label{subsec:PerfBeam}
In this subsection, we investigate the performance of the proposed sub-6GHz based mmWave beam prediction approach. First, we will verify the basic claim that sub-6GHz channel can be directly used to predict the optimal mmWave beams with the aid of deep neural networks. Then, we will evaluate how this prediction performance is affected by the noisy sub-6GHz channel measurements and the mmWave array size. Finally, we will draw an interesting conclusion to the  question: Is it better to predict mmWave beams using sub-6GHz channels or mmWave channels? 

\textbf{Neural networks learn how to predict mmWave beams from sub-6GHz channels:} 
To validate Corollary \ref{cor1} and the capability of deep neural networks in predicting the optimal mmWave beams directly from sub-6GHz channels, we plot the top-1 and top-3 beam prediction accuracy in \figref{fig:training}. In this figure, we adopt the LOS scenario and dataset, described in Sections \ref{subsec:scenarios} and \ref{subsec:dataset} where the noisy channels measured at a 3.5GHz 4-element ULA is used to predict the optimal beam for a 28GHz 64-element array. The user is assumed to use $0$ dBm transmit power for the 3.5GHz uplink pilot (which represents a high-SNR regime of ~$25$dB for the adopted $20$MHz bandwidth and $5$dB noise figure). For this setup,  \figref{fig:training} plots the top-1 and top-3 accuracies versus different training set sizes. The x-axis values indicate the ratio of the training dataset samples that are actually used in training to the total number of training samples. First, \figref{fig:training} confirms the ability of neural networks in predicting the optimal mmWave beams directly from the sub-6GHz channels with high success probability that, for the adopted setup, approaches $85\%$ and $99\%$ for top-1 and top-3 accuracies, respectively. Further, the figure shows that 30\% of the total training subset size  is enough to get a beam prediction success probability $\kappa_1$ that is approximately 12\% off of the upper bound. These results validate the capability of deep neural networks in effectively predicting the mmWave beams using sub-6GHz channels.

\begin{table*}[h]
	\caption{Top-1 and 3 accuracies for sub-6GHz based mmWave beam prediction}
	\centering
	\begin{tabular}{c c c c c c c c}
		\hline\hline
		SNR (dB) & -10 & -5 & 0 & 5 & 10 & 15 & 20 \\
		\hline
		Top-1 & 13.3\% & 26\% & 41.1\% & 57.6\% & 70\% & 78.5\% & 83.1\% \\
		\hline
		Top-3 & 35.9\% & 60\% & 80.4\% & 92.8\% & 96.8\% & 98.3\% & 98.8\% \\
		\hline\hline
	\end{tabular}
	\label{Accuracies}	
\end{table*}

\textbf{Impact of noisy channel measurements at sub-6GHz:} In \figref{fig:training}, we considered a high SNR regime. Now, we want to evaluate the degradation in the mmWave beam prediction performance for different SNR regimes. Note that this SNR refers to the sub-6GHz receive SNR, i.e., how noisy the sub-6GHz channel measurements are. To do that, we considered the same setup of \figref{fig:training} while adding noise with different noise power values to the sub-6GHz channels. Essentially, we study the beam prediction performance for the range of -10dB to 20dB sub-6GHz SNR. For each SNR, the network is trained with the noisy subset of samples, and the Top-1 and Top-3 accuracies are measured on the noisy test subset. The prediction performance at this range is summarized in Table \ref{Accuracies}. As shown in this Table, the proposed deep learning model can clearly combat harsh noise situations. For example, the model can predict the optimal beam within its top-3 predictions with an accuracy close to 81\% at 0 dB SNR, and it is the top-1 prediction around 50\% of the time at the same SNR. This indicates that even in harsh conditions like that, a very little mmWave beam training could refine the network output and improve the performance, i.e., instead of sweeping across the whole codebook (64 beams in this case), the top-3 predictions are 81\% likely to have the best one among them.

To translate this into wireless communication terms, \figref{perf_noise} plots the mmWave achievable rates using the predicted beams for different values of \textit{sub-6GHz} SNR. Note that the upper bound is the same as we consider a fixed mmWave SNR; we are just changing the sub-6GHz SNR. At 0 dB SNR, the top-3 achievable rate is about 6\% shy of the upper bound, which is only achieved with full knowledge of the mmWave channels. The top-1 rate, on the other hand, is not quite as close as the top-3 to the upper bound. It is 39\% off of that bound, yet it is acceptable considering the low SNR. Around 15 dB is where that gap starts closing up, dropping a little less than 5\% for top-1. An important observation needs to be highlighted here. With the Top-1 accuracy at 0 dB a little above 41\% in Table \ref{Accuracies}, it may seem a bit surprising that the rate only drops 39\% from the upper bound. This implies that even when the DNN mis-classifies, it seems to select a beam that is not very far away from the correct one. Such claim is corroborated with the relatively high Top-3 accuracy. This could also be observed at 5 and 10 dB SNRs.

\begin{figure}[t]
	\centering
	\subfigure[]{\includegraphics[width=.47\linewidth]{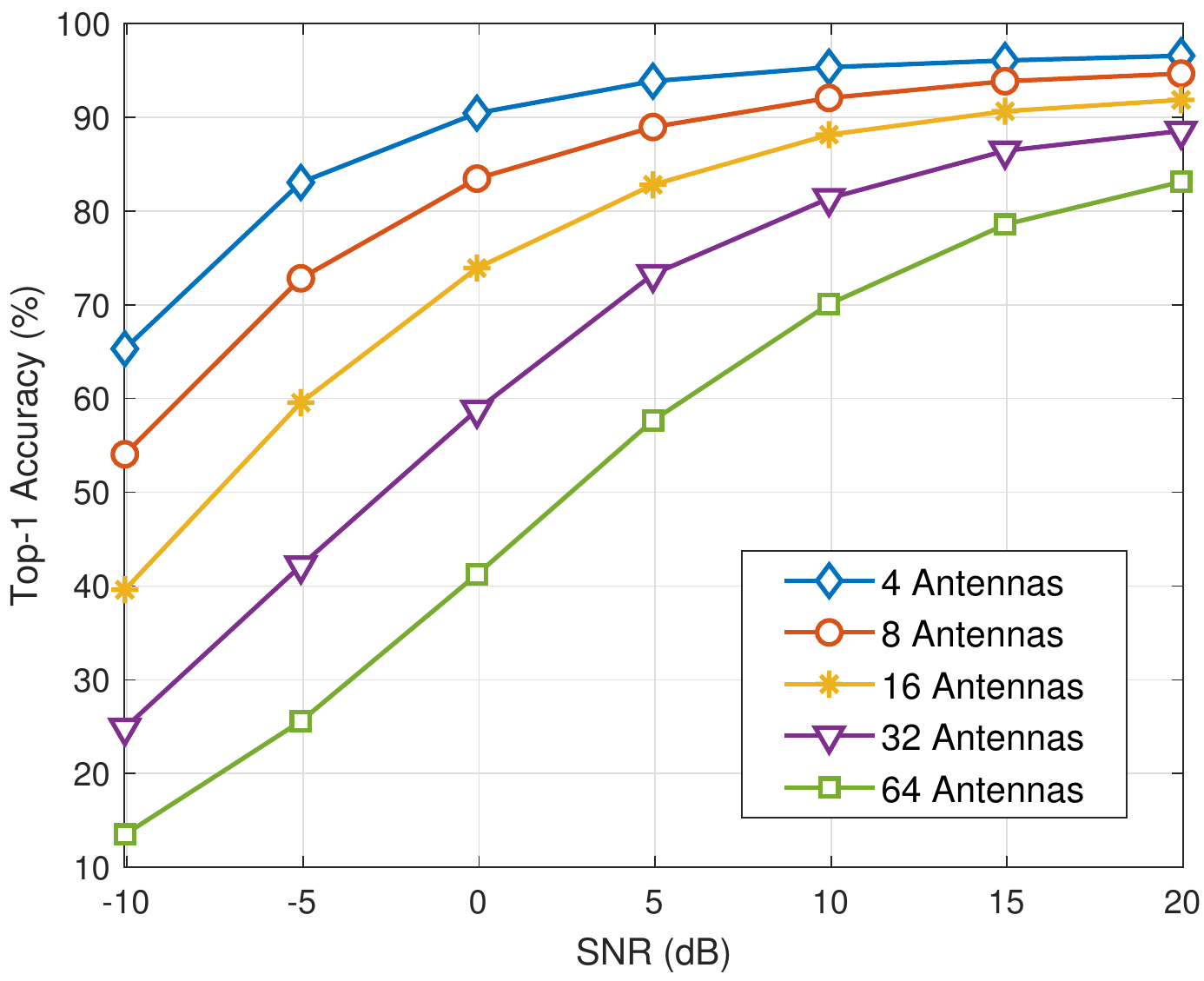} \label{AccVSAnt}}
	\subfigure[]{\includegraphics[width=.47\linewidth]{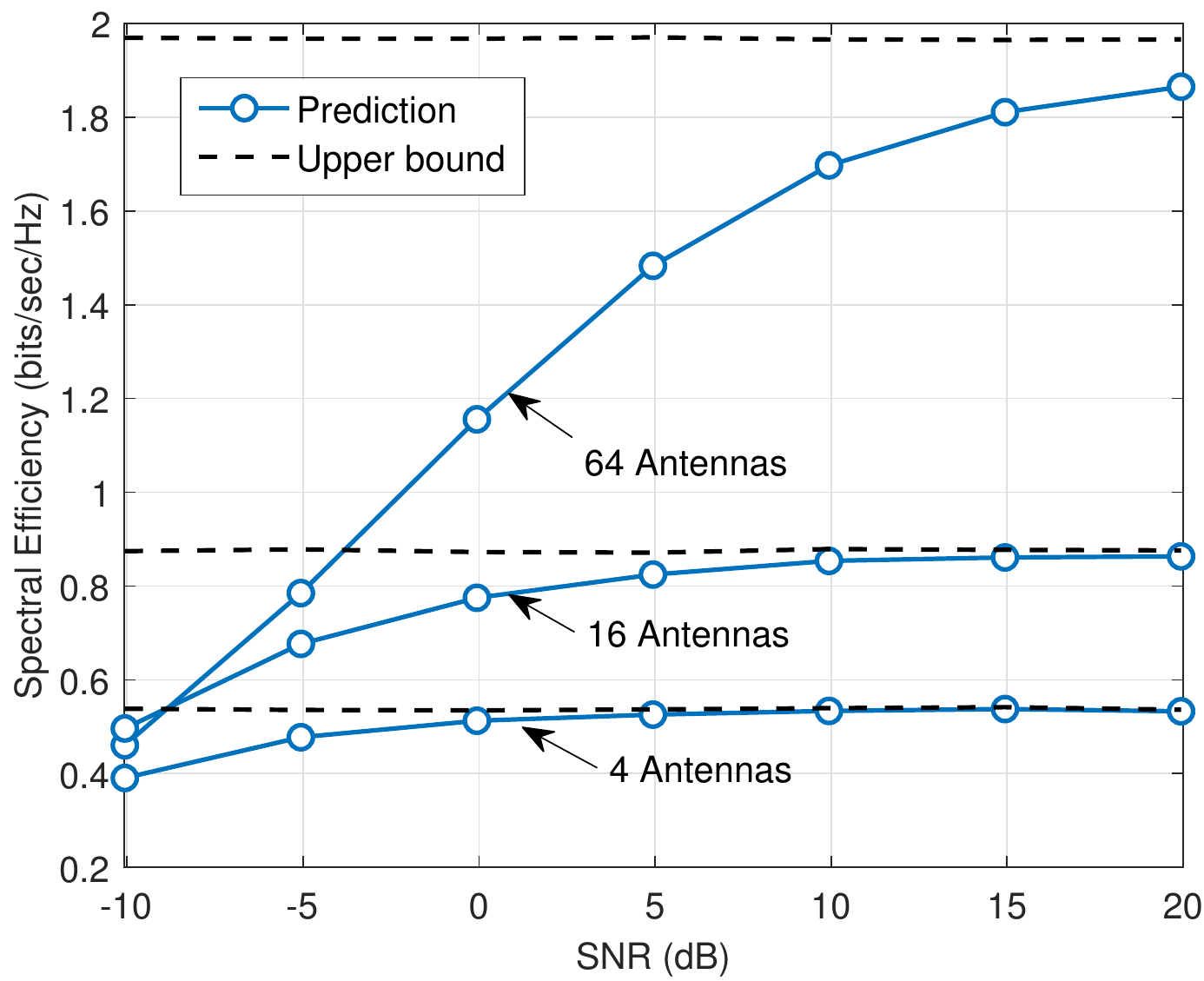} \label{AchVSAnt}}
	\caption{(a) shows the prediction accuracy of the deep neural network as the number of mmWave antenna elements increases. Larger mmWave antenna amounts to larger beam-forming codebook and, therefore, larger number of classes. For some choices of mmWave antennas, the  neural  network model top-1 achievable rate is plotted with its upper bound in (b). All curves are obtained with AWGN only.}
	\label{fig_comb}
\end{figure}

\textbf{Performance with different mmWave array sizes:} With that interesting performance above, one question could come to mind: Is such performance attainable with any number of mmWave antennas? A smaller number of mmWave antennas means there are less classes to learn. On the surface, this looks like an easier prediction task for the deep neural network, which is true. Figure \ref{AccVSAnt} shows the top-1 performance of the neural network with different numbers of mmWave antennas. It is very clear that the proposed deep learning model has better classification performance with a small number of antennas, no matter what the SNR level is. This trend translates to the top-1 achievable rate performance. Figure \ref{AchVSAnt} shows the average achievable rate against SNR for three different mmWave antenna arrays. Although the antenna gain is low with small number of mmWave elements, the deep neural network achieves a much smaller gap with the upper bound for small number of elements compared to that achieved for a large number of elements. This is an immediate reflection of the complexity of the classification task.

\begin{figure}[t]
	\centering
	\includegraphics[width=.5\linewidth]{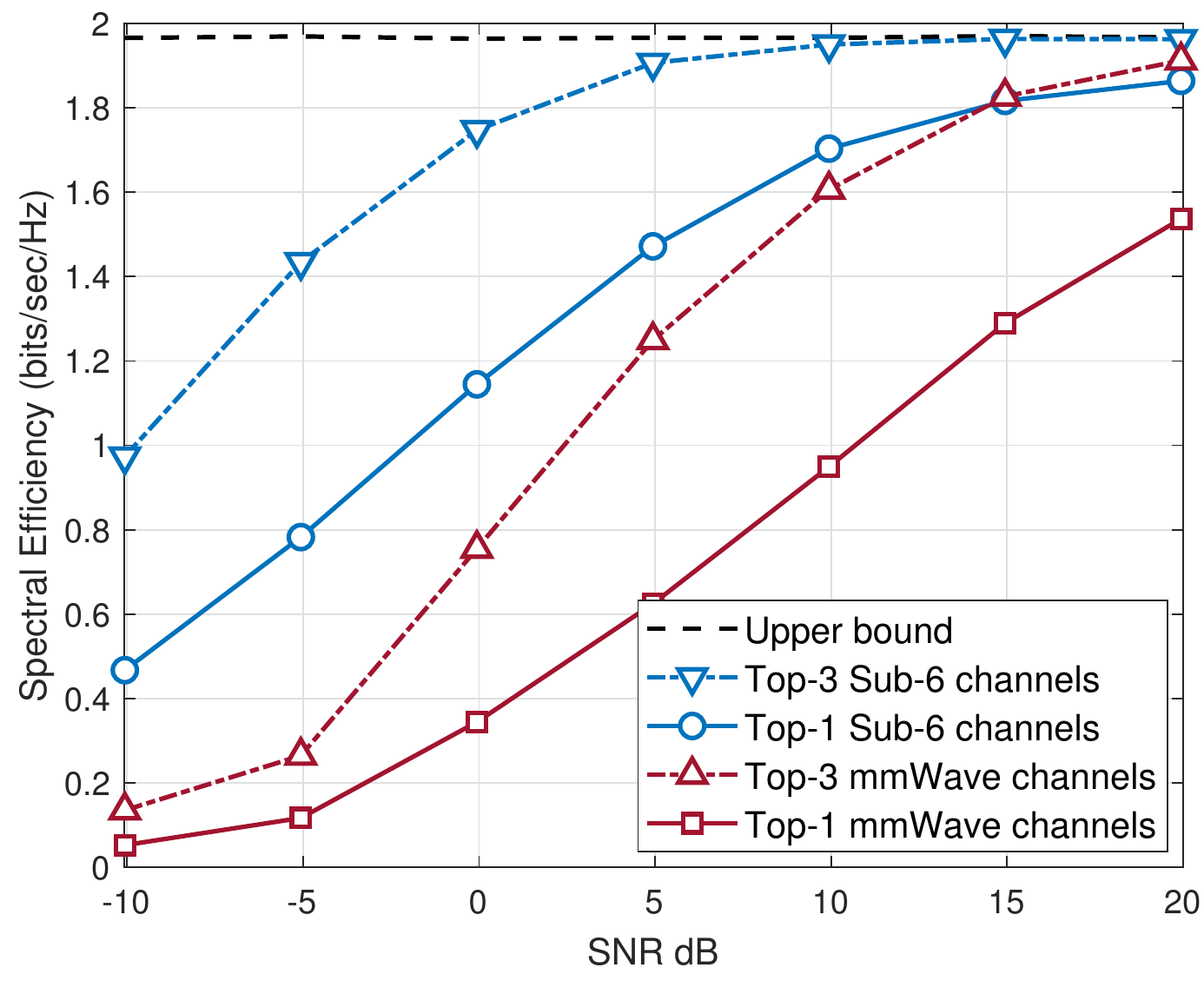}
	\caption{A performance comparison between a neural network trained with inputs from the Sub-6GHz band and another trained with inputs from the mmWave band. The DNN architecture and training hype-parameters used in both cases are the same.}
	\label{mmWaveVSSub6}
\end{figure}

\textbf{Better prediction with mmWave or sub-6GHz channels?} To answer this interesting question, the achievable rates with top-1 and top-3 performance are plotted in \figref{mmWaveVSSub6} for two cases: (i) when the input to the neural network is the sub-6GHz channels or (ii) when the input is the mmWave channel. Even though one would probably expect the mmWave channel knowledge would have a better prediction quality of the mmWave beams, \figref{mmWaveVSSub6} interestingly illustrates that sub-6GHz channels have better beam prediction performance. This is largely because of the large bandwidth of the mmWave channel which results in more corrupted (noisy) channel knowledge compared to the sub-6GHz channel measurements. Such observation emphasizes the promising gain of leveraging sub-6GHz to accurately predict optimal mmWave beamforming vectors.

\subsection{Blockage Prediction} \label{subsec:block_results}
The second set of experiments aims at evaluating the ability of the deep neural networks to \textit{differentiate} blocked and LOS users from the same spatial region. For this end, we adopt the blockage dataset described in \sref{subsec:dataset}, that mixes the LOS and blocked users. Given this dataset, we investigate the blockage prediction performance for the following two labeling approaches
\begin{itemize}
	\item \textbf{Ground-truth labeling:}  This approach assumes the availability of accurate user labels by some means such as, for example, simultaneous localization and mapping techniques. While this may not be a very practical approach, it provides an upper bound for the performance of the other labeling techniques. 
	
	\item \textbf{Power-based labeling:} The LOS paths are normally much stronger (have higher power) compared to the NLOS ones. Therefore, one possible way to differentiate the blocked and unblocked users is by computing the ratio between the power of the strongest beam in the codebook to that of the second strongest beam for each user, referred to as the \textit{power-rule} labelling. This ratio is expected to be large for unblocked users and small (close to one) for blocked users. \figref{fig:hist} corroborates such intuition; it shows two power-ratio histograms, one for the blocked users and the other for LOS users. It is clear that majority of blocked users have power-ratios close to one. With that, a threshold for labeling could be set and used to create the labels during the background training. 
\end{itemize}

\begin{figure}[t]
	\centering
	\subfigure[ ]{\includegraphics[width=.45\linewidth]{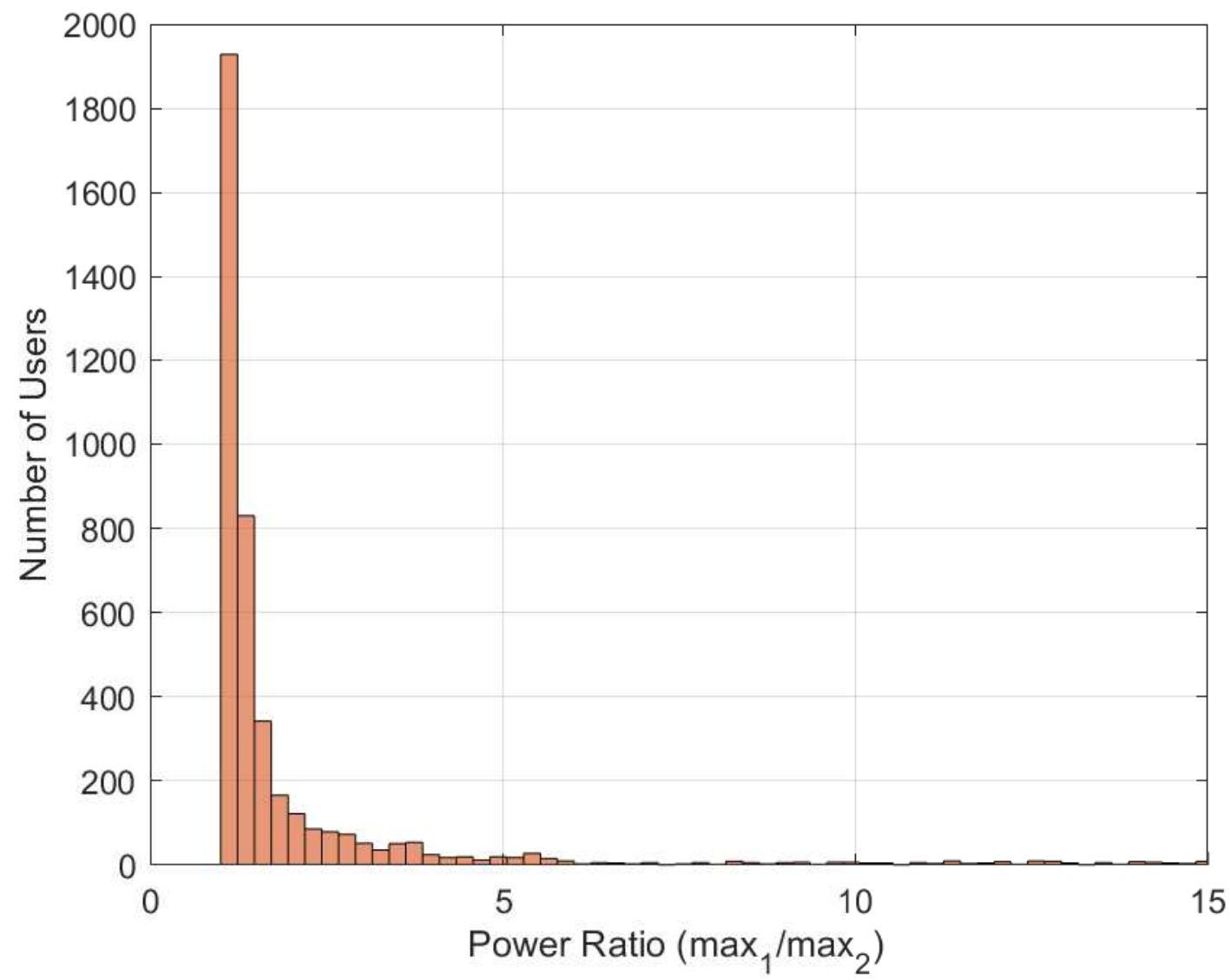}}
	\subfigure[ ]{\includegraphics[width=.44\linewidth]{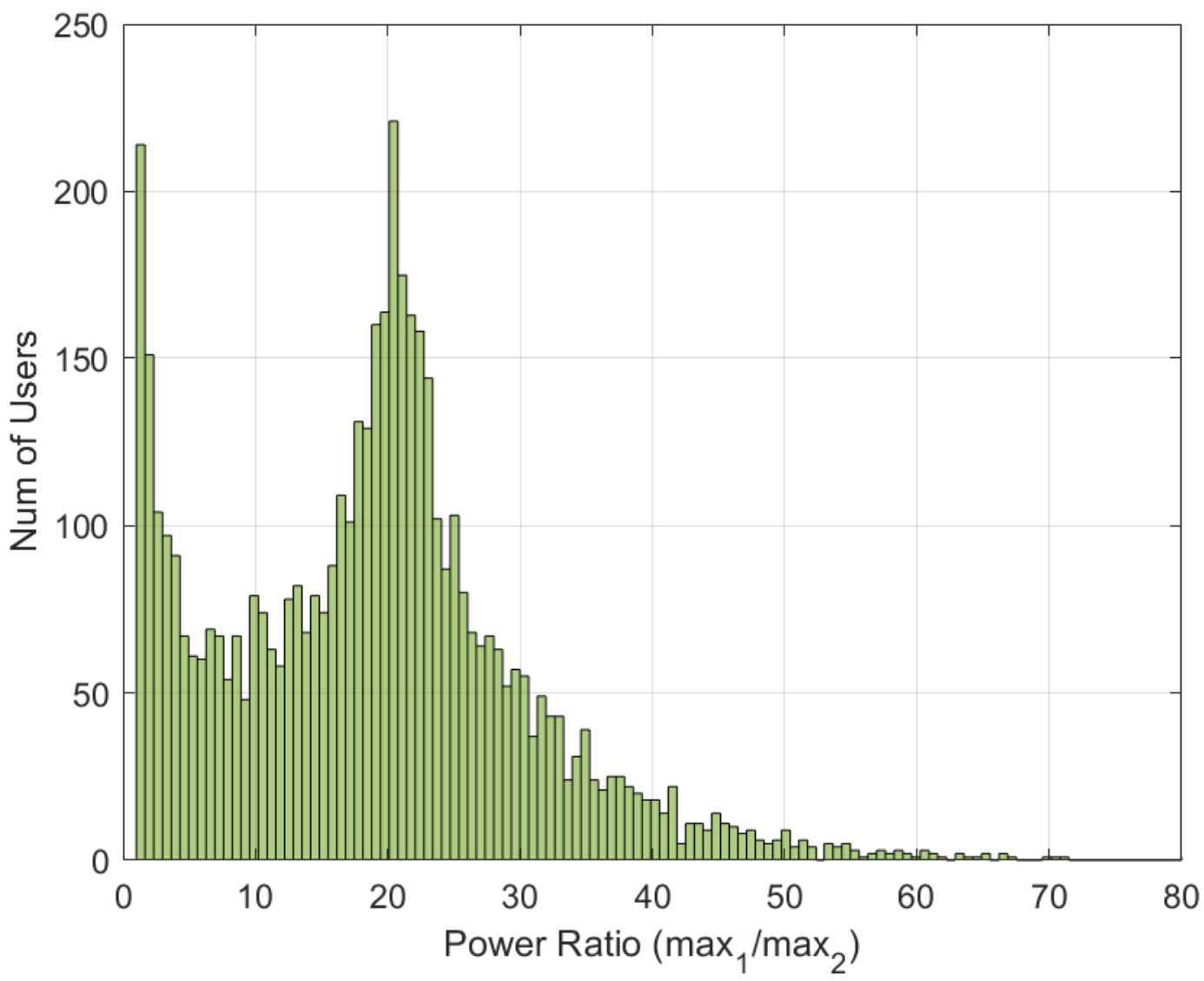}}
	\caption{Power-ratio histograms from the two scenarios considered in this work. (a) is for blocked users while (b) is for LOS users. The ratio measures the difference in power between the strongest beam in the codebook and the second strongest.}
	\label{fig:hist}
\end{figure}

To evaluate the performance of the proposed sub-6GHz based blockage prediction, we consider the blockage scenario and dataset, described in Sections \ref{subsec:scenarios} and \ref{subsec:dataset}, with both the ground-truth and power-rule labeling approaches. For this setup, and as discussed earlier in \sref{subsec:train}, transfer learning is used to train the deep neural network. In \figref{fig:blockage}, we plot the success probability (accuracy percentage) of blockage prediction. First, \figref{fig:blockage} illustrates that the deep learning model has excellent classification ability for the ground-truth labeling approach under a wide range of SNRs. This performance is then compared to the case when the power-rule labeling technique is used. Despite the label contamination, i.e., some miss-labeled users are present during training, the deep neural network model still performs relatively well; its accuracy exceeds 90\% at high SNRs. This highlights the potential of using sub-6GHz channels to effectively predict mmWave link blockages. 
\begin{figure}[t]
	\centering
	\includegraphics[width=.5\linewidth]{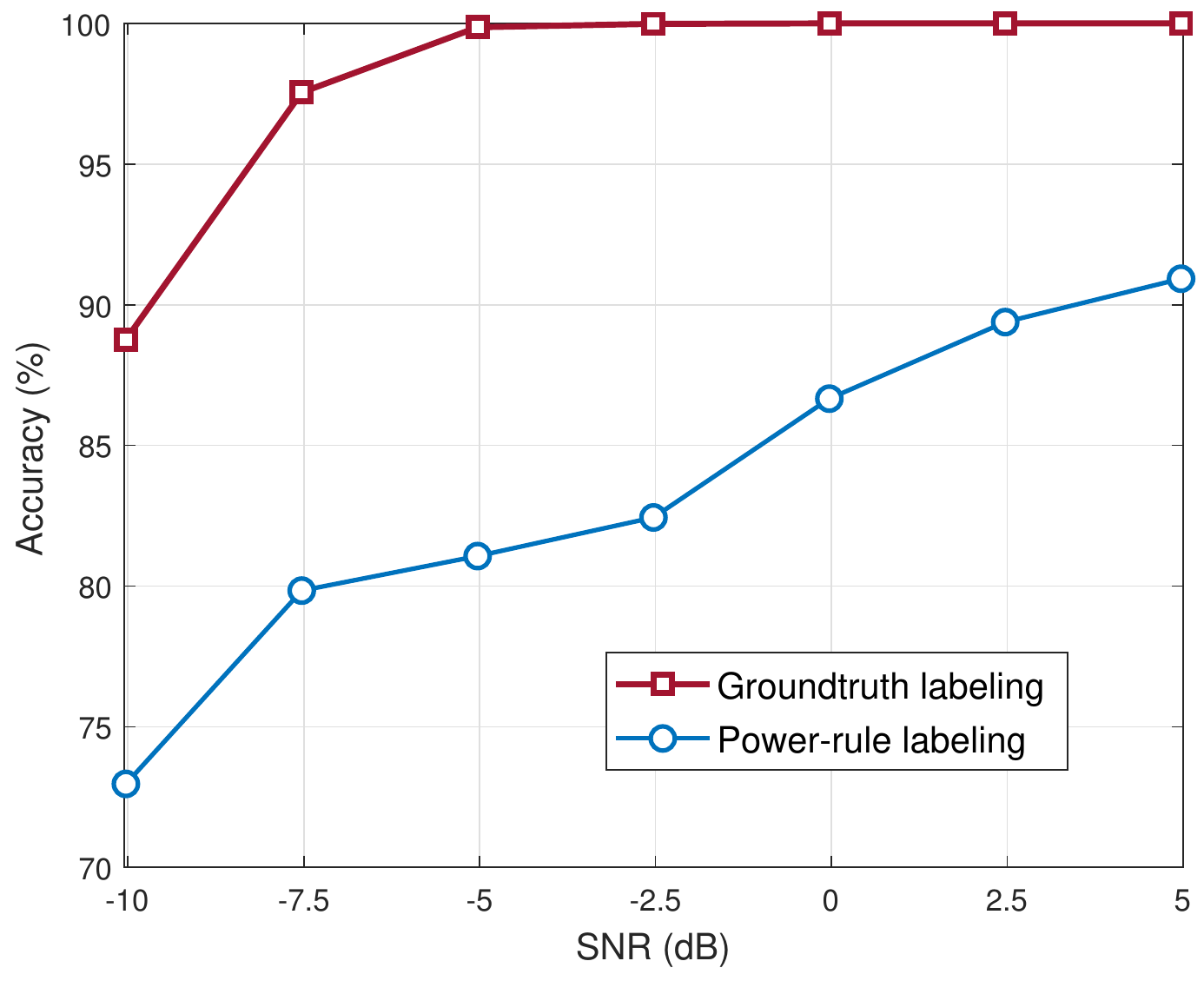}
	\caption{Blockage prediction accuracy of the deep learning model. The best performing model for beam prediction is used for transfer learning. The model was trained with two different label sets, ground-truth and power-rule labeling.}
	\label{fig:blockage}
\end{figure}
\section{Conclusion}
In this paper, we established the conditions under which the mapping functions from a sub-6GHz channel to the optimal mmWave beam and blockage status exist. Leveraging the universal approximation theory, we proved that a large enough neural network can learn these mapping functions such that the success probabilities of predicting the optimal mmWave beam and blockage status be arbitrarily close to one. Therefore, we design a neural network that can perform both prediction tasks using sub-6GHz channels. We use accurate 3D ray-tracing to develop evaluation datasets and test our network. The results show promising and impressive performance; the network, when trained with enough data, does both tasks with relatively high fidelity, even in the presence of noisy sub-6GHz channels. Our beam-prediction experiments reveal an interesting tendency of the network to learn correct beam direction. Although it sometime mis-predicts the mmWave beam, it often selects a beam in the vicinity of the optimal one. This is attainable with small or large mmWave antenna arrays and at reasonable SNRs. Such performance extends to the blockage prediction task; the network, under high SNRs, is capable of predicting the LOS link status with more than 90\% success probability. This could yield interesting gains for the reliability of mmWave systems. For future work, it would interesting to develop learning models that can handle the dynamics of the environment, and design more efficient and practical labeling approaches to label blockage data.

\begin{appendices}
\section{} \label{app1}
\textit{Proof of Corollary \ref{cor1}:} 
The success probability in predicting the optimal mmWave beam $\bff^\star$ using the sub-6GHz channels can be written as
\begin{align}
\kappa_1 & =\mathbb{P}\left(\hat{\bff}=\bff^\star\left|\bh_\text{sub-6}\right.\right) \\
&=1- \mathbb{P}\left(\hat{\bff} \neq \bff^\star\left|\bh_\text{sub-6} \right. \right).
\end{align}
Since the predicted beam $\hat{\bff}$ is obtained from the outputs of the $|\mathcal{F}|$ neural networks by applying  $\hat{n}=\argmax_{n=1,2,...,|\mathcal{F}|} \boldsymbol{\Pi}^n_{N}(.)$ and setting $\hat{\bff}$ as the $\hat{n}$th beam in the codebook $\mathcal{F}$, then $\kappa_1$ can be expressed in terms of $\boldsymbol{\Pi}^n_{N}(.)$ as 
\begin{equation}
\kappa_1= 1- \mathbb{P}\left(      \boldsymbol{\Pi}^{\hat{n}}_N (\bh_\text{sub-6})    >      \boldsymbol{\Pi}^{n^\star}_N (\bh_\text{sub-6})       \right) 
\end{equation}
Now, given Proposition \ref{prop2}, we reach	
\begin{align}
\kappa_1 & \geq 1- \mathbb{P} \left( \boldsymbol{\Phi}^{\hat{n}} + \epsilon_{\hat{n}} > \boldsymbol{\Phi}^{n^\star} -  \epsilon_{n^\star}    \right) \\
& = 1- \mathbb{P} \left( \boldsymbol{\Phi}^{n^\star} -  \boldsymbol{\Phi}^{\hat{n}} <  \epsilon_{n^\star}  +  \epsilon_{\hat{n}}   \right) \\
& \stackrel{(a)}{\geq} 1- \mathbb{P} \left( \boldsymbol{\Phi}^{n^\star} -  \boldsymbol{\Phi}^{\hat{n}} < 2 \bar{\epsilon}   \right)
\end{align}
where (a) follows by defining $\bar{\epsilon}=\max_{n={1,2,...,|\mathcal{F}|}} \epsilon_n$.  Now, given Proposition \ref{prop2} and Assumption \ref{assume_CB}, for any $\epsilon >0$, there exists $\bar{\epsilon}$, such that $\mathbb{P}\left( \boldsymbol{\Phi}^{n^\star} -  \boldsymbol{\Phi}^{\hat{n}} < 2 \bar{\epsilon}_{n^\star}  \right) < \epsilon$, which concludes the proof.

\end{appendices}


\end{document}